\documentclass[twoside,11pt,letterpaper]{article}

\usepackage[english]{babel}
\usepackage{amsmath,amssymb,amsthm}
\usepackage{mathrsfs}
\usepackage[dvips,final]{graphicx}
\usepackage{jair}
\usepackage{theapa}

\usepackage{dsfont}
\usepackage{amsmath}
\usepackage{amssymb}
\usepackage{latexsym}
\usepackage{url}

\newcommand{\w}{\mathit{val}}
\newcommand{\img}{\mathtt{img}}

\newcommand{\game}{\mathcal{G}}
\newcommand{\C}{\mathcal{C}}
\newcommand{\R}{\mathcal{R}}

\renewcommand{\d}{\mathbf{d}}
\newcommand{\val}{\mathit{eval}}

\newcommand{\A}{\mathcal{A}}

\newcommand{\p}{\mathbf{p}}
\renewcommand{\w}{\mathbf{w}}
\renewcommand{\v}{\mathbf{v}}
\renewcommand{\t}{\mathbf{t}}

\renewcommand{\img}{\mathtt{img}}
\renewcommand{\val}{\mathtt{val}}
\newcommand{\opt}{\mathtt{opt}}
\newcommand{\marg}{\mathtt{marg}}
\newcommand{\best}{\mathtt{best}}
\newcommand{\nodes}{\mathtt{agents}}
\newcommand{\dom}{\mathtt{dom}}
\newcommand{\D}{\mathbf{D}}
\newcommand{\norm}{\mathtt{one\mbox{-}good}}
\newcommand{\omeganorm}{\omega\mathtt{\mbox{-}good}}
\newcommand{\mydot}[2]{{#1},{#2}}
\newcommand{\onC}[1]{#1} 
\newcommand{\onCi}[1]{#1} 

\newcommand{\sv}{\phi}
\newcommand{\mymath}[1]{\mathds{#1}}
\newcommand{\hmymath}[1]{\hat{\mathds{#1}}}

\newcommand{\tmarg}{\mbox{{\tiny \it marg}}}
\newcommand{\tbest}{\mbox{\tiny \it best}}

\newcommand{\tuple}[1]{\langle#1\rangle}
\newcommand{\nop}[1]{}

\newcommand{\longv}[1]{}

\newcommand{\products}{\mathit{products}}
\newcommand{\score}{\mathit{score}}
\newcommand{\submitted}{\psi}

\theoremstyle{plain}
\newtheorem{theorem}{Theorem}[section]
\newtheorem{lemma}[theorem]{Lemma}
\newtheorem{proposition}[theorem]{Proposition}
\newtheorem{corollary}[theorem]{Corollary}
\newtheorem{fact}[theorem]{Fact}

\theoremstyle{definition}
\newtheorem{defn}[theorem]{Definition}
\newtheorem{example}[theorem]{Example}

\newenvironment{myitemize}{\begin{list}{$-$}{\setlength{\leftmargin}{13pt}
\setlength{\itemindent}{0.0\labelwidth}}} {\end{list}}

\begin{document}

\title{Mechanisms for Fair Allocation Problems:\\ No-Punishment Payment Rules in Fully Verifiable Settings}


\author{\name Gianluigi Greco \email ggreco@mat.unical.it \\
        \addr Dipartimento di Matematica\\
              Universit\`a della Calabria\\
              I-87036 Rende, Italy
        \AND
        \name Francesco Scarcello \email scarcello@deis.unical.it \\
        \addr DEIS\\
              Universit\`a della Calabria\\
              I-87036 Rende, Italy
}

\maketitle \raggedbottom

\begin{abstract}
Mechanism design is addressed in the context of fair allocations of indivisible goods with monetary compensation.
Motivated by a real-world social choice problem, 
mechanisms with verification are considered in a setting where (\emph{i}) agents' declarations on allocated goods can be fully verified before
payments are performed, and where (\emph{ii}) verification is not used to punish agents whose declarations resulted in incorrect ones. Within this setting, a mechanism is designed
that is shown to be truthful, efficient, and budget-balanced, and where agents' utilities are fairly determined by the Shapley value of
suitable coalitional games. The proposed mechanism is however shown to be {\rm \#P}-complete. Thus, to deal with applications with many agents
involved, two polynomial-time randomized variants are also proposed: one that is still truthful and efficient, and which is approximately
budget-balanced with high probability, and another one that is truthful in expectation, while still budget-balanced and efficient.
\end{abstract}


\section{Introduction}

Whenever the outcome of some social choice process depends on the information collected from a number of self-interested agents, strategic
issues may come into play. Indeed, agents may find convenient to misreport their \emph{types}, i.e., the relevant information they own as their
private knowledge, so that the best possible solution cannot be achieved.
In these cases, \emph{mechanism design} techniques can be used as solution approaches, which augment combinatorial algorithms with appropriate
monetary payments, aimed at motivating all agents to truthfully report their private types~\cite<see, e.g., >{Vazirani2007,Shoham2009}.

On the class of social choice \emph{utilitarian} problems, agent types encode (monetary) valuations over the set of all possible solutions and
the goal is to compute a solution maximizing the \emph{social welfare}, i.e., the sum of agents' true evaluations. A prominent role in
mechanism design for problems of this class is played by the Vickrey-Clarke-Grove (VCG) paradigm~\cite{Vickery1961,Clarke1971,Groves1973},
which is a general method for designing \emph{truthful} mechanisms, i.e., mechanisms where truth-telling is a dominant strategy for each agent.
In particular, VCG mechanisms are \emph{efficient}. That is, they guarantee that a solution maximizing the social welfare is actually computed.
However, they are not \emph{budget-balanced}, i.e., the algebraic sum of the monetary transfers is not zero and mechanisms from this class can
run into deficit. In fact, this is a well-known drawback of VCG mechanisms~\cite<see, e.g., >{Archer2007}, but it is essentially the best one
can hope to do, given classical impossibility theorems~\cite{Green1977,Hurwicz1975}, stating that no truthful mechanism can be designed to be
always efficient and budget-balanced.

In many practical applications, as the one that inspired the present work, payments to agents can be performed \emph{after} the final outcome
is known, so that some kind of \emph{verification} on reported types might be possible. Mechanisms with verification have been introduced by
\citeA{Green1986} and subsequently studied in a number of papers. For instance, \citeA{Nisan2001} considered verification for a task scheduling
problem: We have some agents declaring the amount of time they need to solve each task, and the goal is to have all tasks being solved, by
minimizing the completion time of the last-solved one (hence, the make-span). In this context, payments are computed {after} the actual tasks
release times have been observed, so that we have, for instance, the ability to ``punish'' some agent whose declared ability has been
{verified} to be different than its actual performance in the process.

Note that the knowledge of the actual outcome represents a source of additional information to define payment rules, as it might partially
reveal agents' types. However, this additional power is not considered in the classical mechanism-design setting. In fact, whenever
verification is allowed, some (classical) impossibility results might no longer hold.

Compared to standard mechanisms~\cite<see, e.g.,>{Vazirani2007}, those with verification have received considerably less attention in the
literature \cite<see, e.g., >{Nisan2001,Auletta2009,Penna2009,Krysta2010,Ferrante2009}.
In particular, these works consider a verification ability that is \emph{partial}, in the sense that agents' (mis)reporting is restricted to
true types plus certain lies (e.g., values that are lower than the true ones), and that verification is focused to detect such lies only.
An extension of this basic model has been recently proposed by \citeA{Caragiannis2012}, who assume no a-priori restrictions on the agents'
reported types, within a setting where an agent cheating on her/his type will be caught with some probability that may depend on her/his true
type, the reported type, or both.
In fact, despite these different facets of the verification power, most of the mechanisms with verification proposed in the literature share
the idea of providing incentives to truthfully report private types by exploiting the intimidation of punishing those agents that will be
caught lying.
Moreover, while budget limits have been considered in some approaches \cite<see, e.g.,>{Nisan2001}, no mechanism with verification has been
designed to be budget-balanced, with the focus being on truthfulness and efficiency.

In the paper, we study mechanisms with verification from the same perspective as \citeA{Caragiannis2012}, in particular by considering a
setting where there is no restriction on the agents' reported types. However, differently from this work, we assume that verification is
\emph{deterministic}, i.e., all incorrect declarations on allocated goods will be eventually detected. This leads to a stronger verification
ability, which could be in principle used to easily enforce truthfulness by just punishing those agents whose declarations have been checked to
be incorrect ones.
However, such an approach would be hardly acceptable by agents in real-life applications, because a true punishment would require a clear proof
of a malevolent behavior. Indeed note that, in practice, possible discrepancies between agents' declarations and third-party verified values
may be due to many different reasons, in particular to the subjective judgment of the verifier. Based on this observation and motivated by a
real-world application domain, we will assume in the paper that only a limited use of such verification power can be made, and indeed that
mechanisms have to be designed which are \emph{not based on punishments}---while nonetheless resulting to be truthful, efficient, and even
{budget-balanced}.

\subsection{Mechanisms for Fair Division with Monetary Compensation}

We consider mechanisms with verification in the context of \emph{fair allocation} problems~\cite<see, e.g.,
>{Moulin2003,Young1994,William2011}.
We assume that a set of indivisible goods to be allocated to a set of agents is given. Each agent is equipped with a private preference
relation, which is just encoded as a real-valued function (basically, a monetary valuation) over all possible goods---formal definitions are in
Section~\ref{sec:framework}. An agent can have allocated more then one good, in which case her/his evaluation is additive over them. Moreover,
goods are non-sharable, i.e., each good can be allocated to one agent at most. However, monetary transfers are allowed,  in terms of both
payments charged to agents and monetary compensations provided to them.
The goal is to find an efficient allocation, that is, an allocation maximizing the total value of the allocated goods, by designing rules
guaranteeing that certain desirable properties are achieved, such as truthfulness and \emph{individual rationality}, i.e., no agent is ever
worse off than he would be without participating to the mechanism. Moreover, we want to obtain outcomes that are ``politically'' acceptable.
That is, agents should perceive the designed mechanism as a \emph{fair} one~\cite<see, e.g., >{Brandt2012}, independently of the rules leading
them to be honest. For instance, it is desirable that no agent \emph{envies} the allocation of any another agent, or that the selected outcome
is \emph{Pareto efficient}, i.e., there must be no different allocation such that every agent gets at least the same utility and one of them
even improves.

Note that the above model for fair allocation is general enough to deal with many practical scenarios, and it has been indeed intensively
studied in the literature. One example application is parking space and benefit allocation at a workplace, where each employee gets a parking
space and a share from a fixed benefits package. House allocation problems are another classical example, where agents collectively own a set
of houses, and we look for a systematic way of exclusively assigning a house to each agent, possibly with monetary compensations. A third
example is job allocation among a group of employees, e.g., assigning an unexpected task among the business units of a corporate, where a job
and possibly a monetary compensation is assigned to each employee. Finally, yet another classical application is room assignment-rent division,
where a group of agents shall rent a house, with each of them getting a room and paying a share of the rent of the house.

The model and, in particular, properties of fair allocations with indivisible objects and monetary transfers have been studied, e.g.,
by~\citeA{Svensson1983,Bevia1998,Maskin1987,Tadenuma1993,Meertens2002,Tadenuma1991,Alkan1991,Willson2003,Su1999,Yang2001,Quinzii1984,Sakai2007}.
Moreover, procedures to compute fair allocations have been proposed
by~\citeA{Aragones1995,Klijn2000,Haake2002,Brams2001,Potthoff2002,Abdulkadiroglu2004}.

None of the above listed approaches, however, can guarantee the elicitation of honest preferences from the agents. In fact, the question of
designing truthful and fair mechanisms has been recently considered as well
\cite{Andersson2008,Andersson2009,Svensson2009,Yengin2012,Ohseto2004,Porter2004,Shioura2006}. In these approaches, while budget limits are
sometimes enforced and mechanisms are defined that cannot run into deficit, budget-balance is never guaranteed. Indeed, this comes again with
no surprise, given that no truthful mechanism can be simultaneously fair (e.g., envy-free or Pareto-efficient) and budget-balanced~ \cite<see,
e.g., >{Tadenuma1995,Alcalde1994,Andersson2010}.

To circumvent this impossibility, approaches have been studied that focus on weaker notions of truthfulness. For instance,
\citeA{Andersson2010} and \citeA{Pathak2009} consider a notion of degree of manipulability which can be used to compare the ease of
manipulation in allocation mechanisms, whereas the notion of \emph{weak} strategy-proofness is considered by \citeA{Lindner2010}, i.e.,
cheating agents are always risking an actual loss, and are never guaranteed to cheat successfully.

In this paper, we depart from the settings studied in all such earlier approaches, because we are interested in applications where a form of
verification is available to the mechanism at the time of deciding monetary compensations among agents. It follows that the mechanism may
exploit the knowledge of private agents' types, and hence classical impossibility results may no longer hold.
Indeed, by exploiting this knowledge, and even assuming a setting where agents with verified incorrect declarations are not punished, we can
show that it is possible to design mechanisms for allocation problems that are truthful, {efficient}, budget-balanced, individually rational,
and fair.
Note that having this kind of {\em a-posteriori} knowledge at payment time is quite common to many applications. In fact, we were inspired by
the following real-world scenario, which will be later formalized in Section~\ref{sec:motivation}.

\subsection{The Italian Research Assessment Program (VQR) 2004-2010}

In 2012, the National Agency for the Evaluation of Universities and Research Institutes (ANVUR) has promoted the `VQR' assessment program
devoted to evaluate the quality of the whole Italian research production in the period 2004-2010. Every research structure $R$ has to select
some research products, and submit them to ANVUR. While doing so, the structure $R$ is in competition with all other Italian research
structures, as the outcome of the evaluation will be used to proportionally transfer the funds that have been allocated by the Ministry to
support research activities in the next years (until data from a new evaluation for the subsequent period will be available).

Every structure is then interested in selecting and submitting  its \emph{best} research products.
To this end, the program is articulated in three phases. First, authors are asked to self-evaluate their products, according to some evaluation
criteria defined by groups of experts chosen by ANVUR. Here, it is assumed that, having such criteria, every author is able to perform a
ranking of her/his own products, ideally to equip each of them with a quality score.\footnote{The set of the possible scores is defined in the
VQR guidelines. To our ends, this detail is immaterial and scores are just viewed as (arbitrary) real numbers.} In the second phase, based on
the self-evaluations being collected, every structure $R$ selects and submits to ANVUR (at most) three products for each author affiliated to
$R$,\footnote{Actually, the number of publications is not always three, in some specific exceptions. Again, this is not a relevant issue, as we
shall see in the formalization in Section~\ref{sec:framework}.} in such a way that the sum of the declared, i.e., self-assessed, scores of the
selected products is the maximum possible one for structure $R$, and that each product is formally associated to one author at most.
Finally, ANVUR formulates its independent quality judgment about submitted publications, and the sum of their ``true'' scores (i.e., those
resulting by ANVUR evaluation) is then the VQR score of $R$. Eventually, $R$ will receive funds in the next years proportionally to this score.

In fact, a VQR score is assigned not only to the structure $R$, but also to all its substructures (e.g., to Departments, if $R$ is a
University). Of course, this has an impact on the funds redistribution inside every research structure. Therefore, while it is clear that each
structure has to maximize the total value of the products submitted to ANVUR, whenever the same product has different co-authors, some
strategic issues come into play and co-authors' personal interest may induce them to cheat on the quality of their products, and may lead to
choices going far from the optimum (total) value.
We thus believe that mechanism design is of high practical interest in contexts resembling this one.
In the mechanism, since products are indivisible, i.e., each of them can be formally allocated to one author at most,
payments should reflect ``adjustments'' over the VQR scores, which suitably take into account the various co-authors no matter to whom any
publication is actually allocated. Moreover, as payments can be computed at the end of the process on the basis of the VQR scores made
available by ANVUR, a mechanism with verification can be conceived, with ANVUR playing the role of the verifier.

\subsection{Contributions}

Motivated by the above application scenario, in this paper, we study allocation problems in a strategic setting where agents can misreport
their private types, and we study mechanisms with verification from both the algorithmic and the computational complexity viewpoint.

\medskip

\noindent \textbf{Algorithmic Issues.}
We show that in the given setting none of the classical impossibility theorems discussed above holds. In particular, we exhibit a payment rule
$\p^\xi$ that turns any optimal allocation algorithm, i.e., any algorithm computing an optimal allocation given the reported types, into a
mechanism with verification such that:

\begin{myitemize}
  \item[$\blacktriangleright$] The mechanism is truthful. This is shown by pointing out a number of properties of allocation problems which are of interest on their own.

  \item[$\blacktriangleright$] The mechanism is {efficient}, budget-balanced, individually rational, envy-free, and Pareto efficient.

  \item[$\blacktriangleright$] The mechanism satisfies the following properties, which are in fact crucial---and hence specifically
    discussed---for the case study of the Italian Research Assessment Program:
  \begin{itemize}
    \item The payment rule is indifferent w.r.t.~the scores (possibly cheats) declared for goods that do not occur in the allocation being
        selected (and hence that are not verified). In theory, the feature is irrelevant given that truth-telling is a dominant strategy.
        In practice, this might be rather useful in the VQR program, because unsubmitted research products with their unverified declared
        scores should have no influence on the payments, as they have no influence on the actual score of the research structure after the
        evaluation is carried out. Moreover, it is unlikely that a consensus can be achieved on a mechanism where payments are based on
        unverified self-declarations.

    \item Verification is not used to force truthfulness by just punishing those agents whose reported values are found different from the
        verified ones. Again, the rationale is that a punishment approach would be hardly ``politically'' acceptable in this context
        (charging university professors because of their self-evaluation about some paper would require some convincing proof of their
        malicious behavior, and would start never-ending disputes).

    \item The payment rule is indifferent w.r.t.~the specific optimal allocation being selected by the mechanism. Thus, the score of each
        researcher/substructure is univocally determined by the overall score of the structure, and any researcher/substructure does not
        have any reasonable argument against her/his structure because of possible alternative allocations. In fact, this is a very strong
        kind of fairness property, which immediately entails envy-freeness and Pareto-efficiency.

  \end{itemize}

 \item[$\blacktriangleright$] Agents' utilities are distributed according to the  \emph{Shapley value} of two
  suitably-associated \emph{coalitional games}---see, e.g., \cite{Vazirani2007}, for a comprehensive
  introduction to sharing problems and coalitional games. In fact, the Shapley value is a prototypical solution concept for
  fair division with monetary compensations, and its desirable properties in (games associated with) allocation problems have been largely
  studied in the literature~\cite<e.g., >{Moulin1992,Francois2003,Mishra2007}.

\end{myitemize}

Note that the Shapely value has been studied in mechanism-design contexts too, where emphasis has been given to the pricing problem for a
service provider~\cite{Moulin1997,Moulin1999,Jain2001}: The cost of providing a service is a function of the sets of customers, and the goal is
that of determining which customers (and at a what price) have to receive it.
The model gives rise to a \emph{cross-monotonic} cost-sharing game, where Shapley-value based sharing mechanisms can be defined that are
truthful and budget-balanced, and which achieve the lowest worst-case loss of efficiency over all utility profiles~\cite{Moulin1997}.
With this respect, the pricing rule $\p^\xi$ can abstractly be viewed as a witness that, whenever (partial) verification is possible,
Shapley-valued based mechanisms may also be implemented with no loss of efficiency at all.

\medskip

\noindent \textbf{Complexity Issues.} Computing an optimal allocation on the basis of the reported types is an easy task, which can be carried
out via adaptations of classical matching algorithms. However, one might suspect that our mechanism is not computationally-efficient as it is
based on the computation of a Shapley value. This is indeed a challenging task that involves iterating over all possible subsets of agents.
We analyze these issues, and we provide the following contributions:

\begin{myitemize}
  \item[$\blacktriangleright$] We show that computing the Shapley value for allocation problems is inherently intractable, in
      fact, {\rm \#P}-complete. Note that, while membership results for {\rm \#P} are rather common for  problems  involving Shapley value
      computations, {\rm \#P}-hardness results have only been proven so far  for a few kinds of
      coalitional games, in particular, for \emph{weighted voting games}~\cite{Deng1994} and \emph{minimum spanning-tree games}~\cite{Nagamochi1997}.

  \item[$\blacktriangleright$] In order to deal with scenarios with many agents demanding for computational efficiency, two modified
      rules, $\hat \p^\xi$ and $\bar{\p}^\xi$, are presented, which are founded on a fully polynomial-time randomized approximation scheme for the Shapley value
      computation. The resulting polynomial-time mechanisms retain most of the properties of $\p^\xi$. In particular, the mechanism based on $\hat \p^\xi$ is truthful,
      efficient, and with high-probability it is approximately budget-balanced. Instead, the mechanism based on $\bar \p^\xi$ is truthful in expectation, but
      always efficient and budget-balanced.
\end{myitemize}

\noindent \textbf{Organization.} The rest of the paper is organized as follows. The motivation scenario for the research reported in the paper
is analyzed in Section~\ref{sec:motivation}. Then, Section~\ref{sec:framework} illustrates the formal framework and the basic concepts to
design mechanisms with verification. The payment rule $\p^\xi$ is defined and analyzed in Section~\ref{sec:mechanism}. Connections with
coalitional games are pointed out in Section~\ref{sec:coalitional}, while rules $\hat \p^\xi$ and $\bar \p^\xi$ are defined in
Section~\ref{sec:complexity}, where computational issues  are dealt with. A comparison with some related works is reported in
Section~\ref{sec:related}, while a few final remarks are eventually discussed in Section~\ref{sec:conclusion}.

\section{Motivating Example: The Italian Research Assessment Program (VQR) 2004--2010}\label{sec:motivation}

In this section we describe the motivating example and case study of the present work: the program for evaluating all Italian research
structures for their activities in years 2004--2010. The evaluation is performed by ANVUR, the National Agency for the Evaluation of
Universities and Research Institutes (\url{www.anvur.org}). Substructures (e.g., University departments) will also be evaluated, by considering
the researchers affiliated to them.

\subsection{Formalization}

Let us hereinafter focus on a given structure $R$, let $\R$ be the set of researchers affiliated to $R$ and, for each $r\in \R$, let
$\products(r)$ be the set of the research products of $r$ in the given period 2004-2010. In the first phase of the program, each researcher
$r\in \R$ associates a quality score $\score_r(p)\in \mathbb{R}$ with each product $p\in \products(r)$.

An \emph{allocation} for $R$ is then a function $\submitted$ mapping each researcher $r\in R$ to a set of publications $\submitted(r)\subseteq
\products(r)$ with $|\submitted(r)|\leq 3$ and with $\submitted(r)\cap \submitted(r')=\emptyset$, for each $r'\in \R\setminus\{r\}$. In the
second phase, the structure $R$ computes an \emph{optimal} allocation, i.e., an allocation $\submitted^*$ such that $\sum_{r\in \R}\sum_{p\in
\submitted*(r)}\score_{r}(p)\geq \sum_{r\in \R}\sum_{p\in \submitted(r)}\score_{r}(p)$, for each possible allocation $\submitted$. The goal of
structure $R$ is indeed to maximize its total score (social welfare), and thus the funds that $R$ eventually will receive. Note that this
optimization phase is based on the scores declared by researchers. Therefore, the goal of $R$ will be achieved if authors correctly/truthfully
self-evaluate their products.

\begin{figure}[t]
  \centering
  \includegraphics[width=0.99\textwidth]{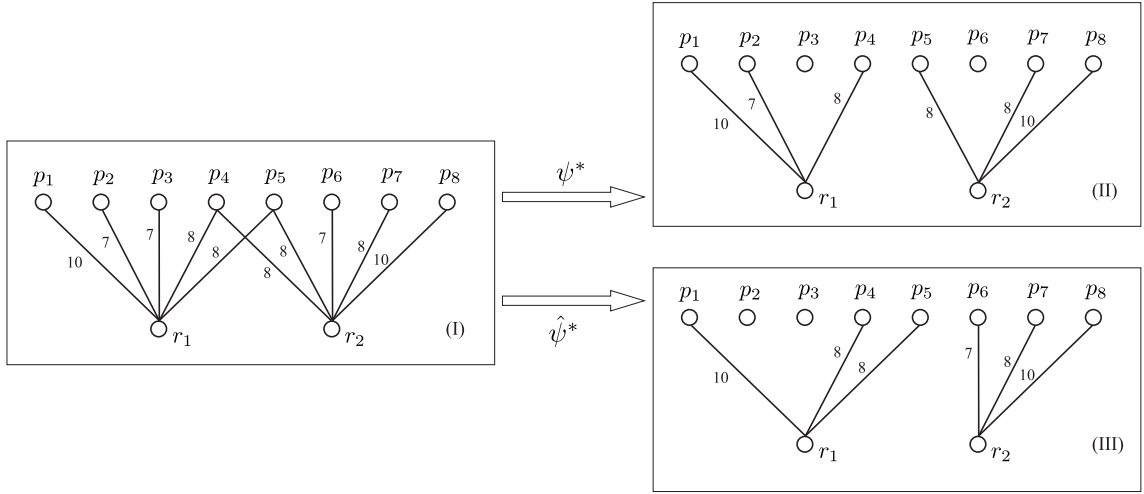}
  \caption{Running example in Section~\ref{sec:motivation}.}
  \label{fig:esempio}
\end{figure}

\begin{example}\label{ex:uno}
Let us consider the simple scenario that is illustrated in Figure~\ref{fig:esempio}(I) by exploiting an intuitive graphical notation. Assume
that there are just two researchers, $r_1$ and $r_2$, affiliated to $R$. Moreover, assume that $\products(r_1)=\{p_1,...,p_5\}$ and
$\products(r_2)=\{p_4,...,p_8\}$, and notice that products $p_4$ and $p_5$ have been co-authored by $r_1$ and $r_2$.
For each $p_i\in \products(r_1)$ (resp., $p_i\in \products(r_2)$), let the self-assessed score $\score_{r_1}(p_i)$ (resp.,
$\score_{r_2}(p_i)$) be the one associated with the edge connecting $r_1$ (resp., $r_2$) to $p_i$.
Given this setting, it is then easily seen that an optimal allocation for $R$ is $\submitted^*$ such that $\submitted^*(r_1)=\{p_1,p_2,p_4\}$
and $\submitted^*(r_2)=\{p_5,p_7,p_8\}$---see Figure~\ref{fig:esempio}(II). \hfill $\lhd$
\end{example}

Let $\submitted^*$ be the optimal allocation being selected by $R$. All products in $\bigcup_{r\in \R}\submitted^*(r)$ are then submitted to
ANVUR for being evaluated (according to the criteria that are publicly available). As a result of the third phase of the VQR program, each
submitted product $p\in \bigcup_{r\in \R}\submitted^*(r)$ is associated with a VQR score $\score_{\mbox{\tiny{VQR}}}(p)\in \mathbb{R}$, which
can be in principle different from the score $\score_r(p)$ declared by some author $r\in \R$. The overall VQR score of the structure $R$,
$\score_{\mbox{\tiny{VQR}}}(\submitted^*)=\sum_{r\in \R}\sum_{p\in \submitted*(r)}\score_{\mbox{\tiny{VQR}}}(p)$, will be eventually translated
to money, proportionally to the performance of $R$ w.r.t.~the performances of the other research structures. Thus, it is actually immaterial to
talk about scores or money in this application scenario.

\begin{example}\label{ex:due}
Assume that there is a precise agreement between self-evaluations (leading to the optimal allocation $\submitted^*$ for structure $R$) and
ANVUR evaluations, i.e., for each researcher $r\in \R$ and each product $p\in \products(r)\cap \submitted^*(r)$,
$\score_{\mbox{\tiny{VQR}}}(p)=\score_r(p)$ holds. Then, we have that $\score_{\mbox{\tiny{VQR}}}(\submitted^*)=51$---see again
Figure~\ref{fig:esempio}(II). Moreover, the funds of $R$ (actually, the part of funds connected to the research performance) will be directly
proportional to $51$.
Therefore, if another structure $R'$ gets 102 as VQR score, its (research-related) funds will be double the funds of $R$. \hfill $\lhd$
\end{example}

\subsection{Division Rules and Desirable Properties}

While the first aim of the VQR program is to evaluate the various Italian research structures, it is known that the obtained information will
be used to evaluate substructures, too  (e.g., University departments). Thus, following the same principle of binding funds to VQR scores used
for the main structures, it is natural to exploit such scores for money distribution inside every research structure (in principle, up to
research groups and individuals).

Of course, this poses the question of how the total score (money) of a structure can be \emph{fairly} distributed over its substructures (and
possibly over individuals), in such a way to reflect their actual contribution to the result achieved by the structure. Formally, we need the
definition of a \emph{division} rule $\gamma$, that is, of a real-valued function that given a researcher $r\in \R$ and an allocation
$\submitted^*$, returns the score $\gamma_r(\submitted^*)$ of $r$ under $\submitted^*$. Then, slightly abusing notation, for any substructure
$\mathcal{S}\subseteq \R$ (here just viewed as the set of its members),  denote by $\gamma_{\mathcal{S}}(\submitted^*)$ the value $\sum_{r\in
\mathcal{S}} \gamma_r(\submitted^*)$.

Surprisingly, no division rule has been formalized in the assessment program, up to now, and this is a source of confusion in many researchers.
As a matter of fact, most researchers believe that the score of substructures will be based on the naive $\tt proj$ rule, where for any
researcher $r$, ${\tt proj}_r(\submitted^*)$ is the sum of the VQR scores of the products allocated to $r$ in $\submitted^*$, i.e., ${\tt
proj}_r(\submitted^*)=\sum_{p\in \submitted*(r)}\score_{\mbox{\tiny{VQR}}}(p)$. For instance, in the setting of Example~\ref{ex:due}, we have
that ${\tt proj}_{r_1}(\submitted^*)=25$ and ${\tt proj}_{r_2}(\submitted^*)=26$.
In fact, on the one hand, this approach trivially satisfies a desirable property of division rules.

\begin{description}
  \item[(P1) ``budget-balance'':] \emph{A division rule $\gamma$ must completely distribute the VQR score of $R$ over all its members,
      i.e., $\sum_{r\in \R}\gamma_r(\submitted^*)=\score_{\mbox{\tiny{VQR}}}(R)$.}
\end{description}

\noindent However, on the other hand, $\tt proj$ might lead to scenarios where some researcher (and in turn some substructure) may have some
reasonable argument against her/his structure because of possible alternative allocations where that researcher may get a higher score. Indeed,
$\tt proj$ does not satisfy the following:

\begin{description}
  \item[(P2) ``fairness'':] \emph{A division rule $\gamma$ should assign to each researcher the best score among all possible allocations,
      i.e., for each $r\in \R$, $\gamma_r(\submitted^*)\geq \gamma_r(\hat{\submitted})$, for any allocation $\hat{\submitted}$.}
\end{description}

\begin{example}\label{ex:tre}
Consider again the optimal allocation ${\submitted}^*$ depicted in Figure~\ref{fig:esempio}(II), and compare it with the allocation
$\hat{\submitted}^*$ of Figure~\ref{fig:esempio}(III). Note that $\hat{\submitted}^*$ is another optimal allocation. Moreover, it is easily
seen that ${\tt proj}_{r_1}(\submitted^*)=25$, whereas ${\tt proj}_{r_1}(\hat{\submitted}^*)=26$. Thus, $r_1$ would complain with her/his
structure, if $\submitted^*$ is selected in place of $\hat{\submitted}^*$.\hfill $\lhd$
\end{example}

In fact, to design a division rule satisfying the above (rather strong) fairness criterium is not straightforward. Indeed, property
\textbf{(P2)} basically tells us that there must be exactly one possible way to distribute the score $\score_{\mbox{\tiny{VQR}}}(R)$ over the
members of $R$, independently of the optimal allocation being actually selected.
Of course, difficulties emerge in those cases where products have multiple co-authors. Thus, to address this issue, one may find natural to
consider and analyze the following two division rules:

\begin{description}
  \item[$\tt owner$:] assign to each author the sum of the ``normalized'' scores of the submitted products (s)he has co-authored, where by
      normalization we just mean here dividing the score of any product by the number of its authors. For instance, in the setting of
      Example~\ref{ex:due}, we have that, in the optimal allocation $\submitted^*$ of Figure~\ref{fig:esempio}(II), half of the score
      associated with $p_4$ (equivalently, $p_5$) is given to $r_1$, and the remaining half to $r_2$. However, even this attempt of having
      a fair division rule is unsuccessful: The setting of Example~\ref{ex:tre} already evidences that this approach does not satisfy
      property \textbf{(P2)}: just check that ${\tt owner}_{r_2}(\submitted^*)=26$ while ${\tt owner}_{r_2}(\hat{\submitted^*})=33$.
      Indeed, according to this division rule, the score of each researcher depends on the number of publications (s)he has coauthored and
      $R$ has submitted to ANVUR, which may be very different in the various allocations.

  \item[$\tt all$:] distribute the score $\score_{\mbox{\tiny{VQR}}}(R)$ to the members of $R$ proportionally to their overall production,
      and not just on the basis of the submitted publications. Accordingly, for any researcher $r\in \R$, we define
      $${\tt all}_{r}(\submitted^*)=\frac{\sum_{p\in \products(r)} \score_r(p)}{\sum_{r\in \R}\sum_{p\in \products(r)}
      \score_r(p)}\times\sum_{p\in \submitted*(r)}\score_{\mbox{\tiny{VQR}}}(p).$$

      In fact, note that $\tt all$ satisfies property \textbf{(P2)}, precisely because the division does not depend on the specific
      products being selected and submitted to ANVUR.
\end{description}

\noindent Actually, we point out here that, while being conceptually simple and still satisfying the fairness requirement, $\tt all$ would be
hardly implementable in practice. First, under $\tt all$, it is clear that researchers might want to act strategically and overestimate the
quality of their own products. Second, even if $\tt all$ is adjusted in a way that it would be always convenient for each researcher to be
truthful, we believe that such a theoretical guarantee would still not suffice to create a consensus on the division rule, given that it
strongly depends (also) on scores that are not certified by ANVUR. Thus, the specific application scenario we are considering suggests the
following additional requirement, which is not satisfied by $\tt all$:

\begin{description}
  \item[(P3) ``implementability'':] \emph{A division rule $\gamma$ must be indifferent w.r.t.~the scores (possibly cheats) declared for
      unverified products, that is, for products not occurring in the selected allocation.}
\end{description}

\noindent So far, none of the proposed examples of division rules emerged to be acceptable. Moreover, in our discussion, we have anticipated
another crucial requirement for a division rule, which is truthfulness. Indeed, depending on the specific division rule being selected,
co-authors may be competitors and might want to act strategically to improve their own score. In particular, while it is clear in principle
that the total value of the products should be maximized (social welfare), authors' personal interest may lead to choices going in a different
(non-optimal) direction. Thus, a division rule must prevent manipulation:

\begin{description}
  \item[(P4) ``truthfulness'':] \emph{A division rule $\gamma$ must provide no incentive in misreporting the score of the research
      products.}
\end{description}

\noindent Very recently, (perhaps) having recognized that the research programme might provide incentive to strategic behaviors, ANVUR pointed
out that only aggregated information about substructures will be made available, rather than the individual scores of the researchers. However,
note that this is not satisfying for two reasons.

First, it is clearly a waste of money to conduct such a thorough evaluation of the quality of the Italian research, without then providing the
output of the results for research products. Indeed, this kind of information would be useful to define the part of the salary of each
researcher that is function of her/his productivity, according to the current law.\footnote{Actually, this applies only to tenured positions at
Universities.} This is so evident that many researchers still believe that such an information will be used for their personal evaluation (soon
or later), and thus are adopting strategic behaviors to have allocated the best products (usually, under the assumption that the rule $\tt
proj$ will be used). However, this might not lead to the global optimum for their research structure, because of possible poor performances of
co-authors, as we shall see in the example below.

Second, disclosing only aggregated information does not prevent at all the emergence of strategic behaviors. Indeed, such strategic issues
still emerge as soon as two researchers from different substructures co-authored some research product, with each of them being interested in
providing as much as contribution as possible to her/his own substructure. Again,  this might not lead to the global optimum, as we next
exemplify for the rule $\tt proj$.

\begin{figure}[t]
  \centering
  \includegraphics[width=0.9\textwidth]{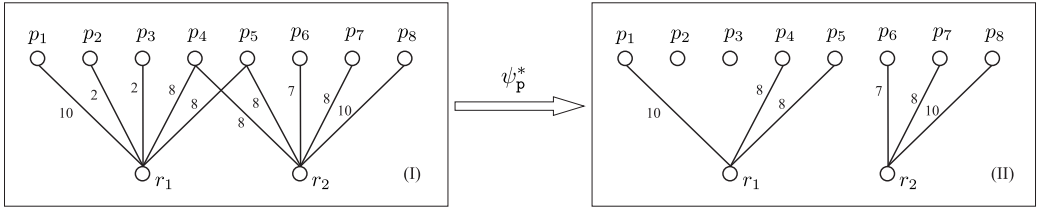}
  \caption{Strategic manipulations with the rule $\tt proj$.}
  \label{fig:esempio3}
\end{figure}

\begin{example}\label{ex:proj}
Assume that $r_1$ and $r_2$ belong to different substructures. Consider the rule $\tt proj$, and assume that researcher $r_1$ declares that
her/his products $p_2$ and $p_3$ are of poor quality (e.g., $\score_{r_1}(p_2)=\score_{r_1}(p_3)=2$), as it is illustrated in
Figure~\ref{fig:esempio3}(I). Then, an optimal allocation $\submitted^*_{\tt p}$ is the one shown in Figure~\ref{fig:esempio3}(II), where the
set $\{p_1,p_4,p_5,p_6,p_7,p_8\}$ of products is submitted to ANVUR. Assume that, for all these products, there is an agreement between
declared scores and ANVUR ones. It follows that ${\tt proj}_{r_1}(\submitted^*_{\tt p})=26$.
On the other hand, recall that in the allocation $\submitted^*$ of Figure~\ref{fig:esempio}(II), which has been computed based on the
declaration that $\score_{r_1}(p_2)=score_{r_1}(p_2)=7$, it holds that ${\tt proj}_{r_1}(\submitted^*)=25$. Thus, $r_1$ finds convenient to
misreport the true scores of $p_2$ and $p_3$, and underestimate them. Note however that the overall score of the structure $R$ is still $51$
and, in fact, $\submitted^*_{\tt p}$ coincides with the optimal allocation $\hat{\submitted}^*$ depicted in Figure~\ref{fig:esempio}(III) and
discussed in Example~\ref{ex:tre}.

Then, consider a slight variation of the problem instance depicted in Figure~\ref{fig:esempio3} where the actual value of product $p_7$ is $6$
(instead of $8$). Then, the above egoistic behavior of agent $r_1$ also damages its research structure because it leads to a sub-optimal
allocation. Indeed, due to the low declared values for $p_2$ and $p_3$, product $p_7$ is selected and allocated to $r_2$ in the unique (wrong)
optimal allocation, whose total score is now $49$ (instead of $51$).\hfill $\lhd$
\end{example}

\begin{figure}[t]
  \centering
  \includegraphics[width=0.9\textwidth]{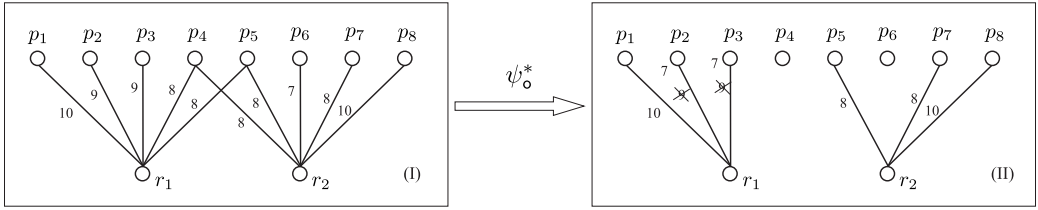}
  \caption{Strategic manipulations with the rule $\tt owner$.}
  \label{fig:esempio2}
\end{figure}

Analogously, the emergence of strategic issues with the rule $\tt owner$ can be easily seen.

\begin{example}\label{ex:quattro}
Assume again that $r_1$ and $r_2$ belong to different substructures. Consider now the rule $\tt owner$, and assume that researcher $r_1$
declares that $\score_{r_1}(p_2)=\score_{r_1}(p_3)=9$. Then, consider the optimal allocation $\submitted^*_{\tt o}$ shown in
Figure~\ref{fig:esempio2}(II), and note that, in this case, the set $\{p_1,p_2,p_3,p_5,p_7,p_8\}$ of products is submitted to ANVUR.

Now, assume that the VQR scores of these products are those illustrated in Figure~\ref{fig:esempio}(I), i.e., the same ones as those discussed
in Example~\ref{ex:uno}. Thus, $r_1$ has cheated with the aim of overestimating the products of which (s)he is the sole author. In fact, this
is convenient to her/him, since according to the rule $\tt owner$, $r_1$ now gets  ${\tt owner}_{r_1}(\submitted^*_{\tt o})=24+4$, because of
the products in $\{p_1,p_2,p_3\}$ allocated to her/him (with overall score 24) and of the product $p_5$ co-authored with $r_2$ (whose overall
score is then shared with $r_2$). On the other hand, just recall that in the allocation $\submitted^*$ of Figure~\ref{fig:esempio}(II), which
has been computed based on the ``truthful'' declaration that $\score_{r_1}(p_2)=\score_{r_1}(p_3)=7$, it holds that ${\tt
owner}_{r_1}(\submitted^*)=25$. It follows that $r_1$ finds convenient to cheat under $\tt owner$, in order to increment the number of products
submitted to the VQR that (s)he has coauthored. However, the egoistic behavior of agent $r_1$ again damages its research structure, as we now
have that the total VQR score is 50 (instead of 51)---see again Figure~\ref{fig:esempio2}(II). \hfill $\lhd$
\end{example}

As a matter of fact, the emergence of strategic issues across substructures risks to penalize, in the long term period, collaborations and
cross-fertilizations.\footnote{It is not by chance that the authors of this paper belong to different substructures of the same University.} To
prevent all these problems, a truthful mechanism is of course definitively needed. In fact, as the above examples might have already suggested
to the careful reader, truthfulness can be achieved by exploiting the fact that ANVUR evaluation corresponds to a ``verification ability''
available in the setting. For instance, one might punish (e.g., by assigning 0 as overall score) all those researchers whose reported values
are found different from the verified ones (usually interpreted as ``lying agents''), as it is in the spirit of most of the literature on
mechanisms with verification~\cite<see, e.g.,
>{Auletta2009,Penna2009,Krysta2010,Ferrante2009,Caragiannis2012}. Indeed, under the intimidation of a punishment, any (reasonable)
division rule can be turned into a truthful one. However, in the application scenario we are considering,
a punishment approach would be hardly ``politically'' acceptable---just think that a number of researchers have already announced that they
will not participate to the VQR program because they disagree with some of the evaluation criteria made available by ANVUR, which are in fact
perceived as imposed by ``law'' rather than as being the outcome of a public discussion on the subject. Moreover, charging researchers because
of some discrepancy between their self-evaluation about some paper and the one by ANVUR experts would require some convincing proof of their
malicious behavior. Therefore, any punishing approach would be quite hard to implement in practice, for this real-world case study.
For this reason, we avoid this brute-force approach, and ask that the following property holds.

\begin{description}
  \item[(P5) ``no punishment'':] \emph{A division rule $\gamma$ must be such that, for each $r\in \R$ and each allocation $\submitted^*$,
      the value $\gamma_r(\submitted^*)$ is indifferent w.r.t.~self-assessed scores, in particular, w.r.t.~discrepancies possibly emerging
      between such scores and VQR ones.}
\end{description}

\noindent Note that, in the light of the above requirement, we look for a method to enforce truthfulness where verification is used in a rather
limited sense. Moreover, it is relevant to observe that if the division rule is not well-designed, then cases might emerge where there is no
way at all to exploit verification (even in its strongest form where punishment is allowed).

\begin{example}
Consider again the use of the rule $\tt proj$ in Example~\ref{ex:proj}, and recall that $r_1$ finds convenient to underestimate the true scores
of $p_2$ and $p_3$. However, since $p_2$ and $p_3$ does not occur in $\submitted^*_{\tt p}$, as we can see in Figure~\ref{fig:esempio3}(II),
then there is no way to discover that $r_1$ has actually cheated. Therefore, in this case, verification on the selected products provides
no-extra power, and truthfulness is not achieved. \hfill $\lhd$
\end{example}

\section{Fair Allocation Problems}\label{sec:framework}

By looking at the examples discussed in the above section, it emerges that defining a ``good'' division rule is not an easy task in allocation
problems with indivisible goods.

In this section, we provide the formal framework for studying such problems, based on mechanism design tools. In particular, we focus on
mechanisms equipped with a verification ability, which meets the ``no-punishment'' perspective. To help the intuition, after notions are
defined in the formal framework, we discuss how they fit our running (motivating) example about the Italian research assessment programme.

\subsection{Formal Framework: Allocations and Strategic Setting}

Assume that a universe $U$ of \emph{indivisible} goods is given.
An \emph{allocation scenario} over $U$ is a tuple $S=\tuple{\A,G,\omega}$ where $\A=\{1,...,n\}$ is a set of agents, $G\subseteq U$ is a set of
goods, and $\omega:\A\mapsto \mathbb{N}$ is a function associating each agent $i\in \A$ with a natural number $\omega(i)>0$.

An \emph{allocation} for $S$ is a function $\pi$ mapping each agent $i\in \A$ into a set $\pi(i)\subseteq G$ of goods with $|\pi(i)|\leq
\omega(i)$ and such that $\pi(i)\cap\pi(j)=\emptyset$, for each agent $j\neq i$. Note that $\omega$ provides the upper bounds on the number of
goods that can be allocated to agents. Of course, this is more general than assuming that no bound is given, which is just the special case
$\omega(i)=|G|$, for each agent $i\in\A$. Eventually, note that goods cannot be shared.
In the following, we denote by $\dom(\pi)$ the domain of $\pi$, and by $\img(\pi)$ the set of goods occurring in the image of $\pi$, that is,
$\bigcup_{i\in \dom(\pi)} \pi(i)$.

The \emph{type} of an agent is a real-valued function over $G$, which is meant to express her/his evaluation of each single good. A vector
$\w=(w_1,...,w_n)$ where $w_i:G\mapsto \mathbb{R}$ is the type of agent $i$, for each $i\in\A$, is called a type vector. The \emph{value} of an
allocation $\pi$ for $S$ w.r.t.~$\w$, denoted by $\val{(\pi,\w)}$, is given by $\sum_{i\in \dom(\pi)} w_i(\pi)$. Then, $\pi$ is \emph{optimal}
(for $S$ w.r.t.~$\w$) if there is no allocation $\pi'$ for $S$ such that $\val(\pi',\w)>\val(\pi,\w)$. The value of an optimal allocation for
$S$ w.r.t.~$\w$ is denoted by $\opt(S,\w)$, while the set of all optimal allocations for $S$ w.r.t.~$\w$ is denoted by $\Pi^*_{S,\w}$, or
simply by $\Pi^*_\w$ if $S$ is understood from the context.

For each agent $i$, let $D_i$ be the set of all her/his possible types, and let $\D$ be the Cartesian product $D_1\times\cdots\times D_n$.
Hereinafter, we assume that an allocation scenario $S=\tuple{\A,G,\omega}$ is given, together with two type vectors $\t=(t_1,...,t_n)\in\D$ and
$\d=(d_1,...,d_n)\in \D$. In particular, we assume that each agent $i$ {reports} a \emph{declared} type $d_i$, which might be different from
her/his actual evaluation $t_i$ of the available goods, called the {\em true} type of $i$. Vector $\d$ is public knowledge, while vector $\t$
is not, because type $t_i$ is private knowledge of agent $i$.

\medskip

\noindent \textbf{Allocation problems $\mapsto$ VQR.} The motivating scenario discussed in Section~\ref{sec:motivation} is a special case of
the above general framework. Indeed, we can associate each structure $R$ with an allocation scenario $S_R=\tuple{\A,G,\mathbf{3}}$, where $\A$
is the set of researchers affiliated to $R$, $G$ is the set of their products, and $\mathbf{3}:\A\mapsto \{3\}$ is the constant function
stating that each researcher can submit 3 products at most. In this setting, $\d$ encodes the scores declared by the researchers (with the
negative score $-1$ being conventionally assigned by $d_i$ to the products that $i$ declares not to have authored). This correspondence is
exemplified below.

\begin{example}\label{ex:ANVURallocation}
Consider again the setting of Example~\ref{ex:quattro}. The given research structure $R$ can be modeled as the allocation scenario
$\tuple{\{r_1,r_2\},\{p_1,...,p_8\},\mathbf{3}}$. Moreover, the vector $\d$ is such that, for each $i\in \{1,2\}$ and $j\in\{1,...,8\}$:
$d_{r_i}(p_j)=\score_{r_i}(p_j)$, if $p_j\in \products(r_i)$, and $d_{r_i}(p_j)=-1$, otherwise.
Note that Figure~\ref{fig:esempio2}(I) can be viewed as a graphical representation of $\tuple{\{r_1,r_2\},\{p_1,...,p_8\},\mathbf{3}}$ where
negative edges are omitted. For instance, for researcher $r_1$, we have $d_{r_1}(p_1)=10$, $d_{r_1}(p_2)=9$, $d_{r_1}(p_7)=-1$, and so on.
Clearly, optimal allocations for the research assessment problem one-to-one correspond to optimal ones for the corresponding allocation
scenario w.r.t.~$\d$. For instance, the allocation $\submitted_o^*$ in Figure~\ref{fig:esempio2}(II) is an optimal allocation for
$\tuple{\{r_1,r_2\},\{p_1,...,p_8\},\mathbf{3}}$ w.r.t.~$\d$.  \hfill $\lhd$
\end{example}

Concerning the vector $\t$ of true types in the VQR setting, recall that research products are evaluated by ANVUR according to some publicly
available criteria, and that, having such criteria, every author is expected to be able of  self-evaluating her/his own products in a way
consistent with ANVUR evaluation. Thus, the vector $\t$ encodes the ``true score'' of each product, as it can be determined according to ANVUR
evaluation criteria.

Actually, note that in the case of the VQR setting, different researchers cannot have truly different true types w.r.t.~a product they have
co-authored, because we are assuming that true types coincide with ANVUR evaluations and ANVUR provides just one value for each product. More
formally, in any VQR allocation problem $\tuple{\A,G,\mathbf{3}}$, the underlying vector $\t$ of true types is such that $t_i(g)>0$ and
$t_j(g)>0$ implies that $t_i(g)=t_j(g)$, for each good $g\in G$ and pair of agents $i,j\in \A$.
Of course, since the proposed framework is more general, it may well support possible extensions of the current VQR setting tailored to a
finer-grained analysis of interdisciplinary products. Indeed, for such products, co-authors can actually have different valuations for the same
paper (e.g., for they belong to scientific communities with a different research focus) and, therefore, it makes sense to have the product
reviewed by a different panel of experts for each author/area of interest.

\subsection{Mechanisms with Verification}

Throughout the paper, we consider mechanism design in a setting where a third-party, called the \emph{verifier}, is formalized as a function
$\v$ associating any allocation $\pi$ with a vector $\v(\pi)=(v_1,...,v_n)$ of \emph{verified} types, with $v_i:\img(\pi)\mapsto \mathbb{R}$,
for each agent $i\in \A$.
In particular we assume that, for each agent $i\in \A$ and good $g\in \img(\pi)$, $v_i(g)=t_i(g)$. Note that a verifier is always able to
\emph{precisely} determine the real value (agents' type) $t_i(g)$ of any good $g$ in $\img(\pi)$, i.e., of any good allocated to some agent by
$\pi$.
Thus, in our setting, values for allocated products can be fully verified.
Instead, $v_i$ is undefined over values that are not allocated, and hence not verified.

In order to encourage agents to truthfully report their private types, we shall design mechanisms where monetary transfers can be performed,
after verified types are available.
Formally, a \emph{payment rule} $\p=(p_1,...,p_n)$ is defined as a vector of functions,
with $p_i(\pi,\d)$ being some amount of money that is given to agent $i$, on the basis of an allocation $\pi$ and a vector $\d$. Observe that,
with this notation, any negative value $p_i(\pi,\d)$ means that some amount of money is charged to agent $i$.
Then, $i$'s (quasi-linear) \emph{utility} under $\p$, sometimes called \emph{individual welfare}, is defined as
$u_{i,{\p}}(\pi,\d)=v_i(\pi)+p_i(\pi,\d)$.
Whenever the payment rule is easily understood from the context, $i$'s utility is simply denoted as $u_i(\pi,\d)$.
Note that utilities depend on verified types, as in the setting by \citeA{Nisan2001}. In fact, we could have equivalently stated them in terms
of the true types. However, the proposed formulation best reflects the intended meaning of the framework, where payments---and hence
utilities---depend on the verification process rather than on personal beliefs of agents.

An \emph{allocation algorithm} is a function $A:\D\mapsto \Pi$ mapping each vector $\w\in \D$ into an allocation $A(\w)$. We say that the
algorithm $A$ is {\em optimal} if $A(\w)\in \Pi^*_\w$, for each $\w\in \D$.

A \emph{mechanism with verification} is a pair $(A,\p)$, where $A$ is an allocation algorithm and $\p$ is a payment rule that can exploit the
power of a verifier $\v$. The mechanism $(A,\p)$ can be viewed as consisting of the following two-phases: First, agents report a declaration
vector $\d$, and an allocation $\pi=A(\d)$ is computed, by using some allocation algorithm $A$. Second, the true types ``restricted'' over
$\img(A(\d))$ are revealed, i.e., $\v(\pi)$ is made available, and payments under a given rule $\p$ are calculated for $A(\d)$ and $\d$, by
exploiting the knowledge of $\v(\pi)$.
Intuitively, our goal is to design a payment rule $\p$ guaranteeing that declared types in $\d$ induce an allocation $A(\d)$ maximizing the
social welfare, i.e., such that $A(\d)\in \Pi^*_\t$ holds.
A comparison of our approach to verification with existing ones is reported in Section~\ref{sec:related}.

\medskip

\noindent \textbf{Allocation problems $\mapsto$ VQR.} In the VQR setting, we have already pointed out that ANVUR precisely plays the role of a
verifier. For instance, by considering again Example~\ref{ex:quattro}, we have that the set $\{p_1,p_2,p_3,p_5,p_7,p_8\}$ is submitted to ANVUR
and hence, for instance, $v_{r_1}(p_1)=t_{r_1}(p_1)$, because product $p_1$ of $r_1$ has been verified and hence $r_1$'s evaluation about it
has been disclosed---see also Figure~\ref{fig:esempio2}(II).
Moreover, negative values model the types of researchers for products that cannot be assigned to them. E.g., we have
$v_{r_1}(p_8)=t_{r_1}(p_8)=-1$, because $r_1$ is not an author of $p_8$, which is always trivially checked by the mechanism (note that the
specific negative value is immaterial).

Our goal is then to single out a mechanism with verification that uses appropriate ``monetary'' transfers to compensate the (globally) optimal
allocation of the given indivisible goods. In particular, note that, in order to be consistent with the quasi-linear setting, payments in the
VQR setting have to be intended as redistributions of the VQR overall score $\score_{\mbox{\tiny{VQR}}}(R)$. As a result, the role of a
division rule is now to determine the individual welfare $u_{i}(\pi,\d)$ associated with each researcher/agent $i$. %
In fact, payment rules should guarantee that welfare values satisfy a number of desirable properties, in particular, the relevant ones pointed
out in Section~\ref{sec:motivation}.

As a further remark, note that our choice of making the dependence on verified types explicit in the definition of the utility function is
conceptually clearer than assuming a dependence on true types (in fact, eventually coinciding with verified ones). Indeed, our choice is
adherent to the VQR setting, where funds (the actual utility of each researcher) will be determined according to ANVUR evaluations.

\subsection{Properties of Mechanisms}\label{sec:properties}

For any vector $\w=(w_1,...,w_n)\in \D$ of types (declarations) and for any type $\bar w_i\in D_i$, we denote by $(\mydot{\bar w_i}{\w_{-i}})$
the type vector $(w_1,...,w_{i-1},\bar w_i,w_{i+1},...,w_n)\in \D$.

Let $(A,\p)$ be a mechanism with verification, and let $i$ by any agent in $\A$ and $d_i$ her/his declared type. We say that $d_i$ is a {\em
dominant strategy} of agent $i$ w.r.t. $(A,\p)$ if, for each vector $\w\in \D$, $u_{i}(A(\mydot{d_i}{\w_{-i}}),(\mydot{d_i}{\w_{-i}}))\geq
u_{i}(A(\w),\w)$ holds.

The mechanism $(A,\p)$ is \emph{truthful} if, for each $i\in \A$, $t_i$ is a dominant strategy. This is in fact property \textbf{(P4)} in
Section~\ref{sec:motivation}.
Let $(A,\p)$ be any truthful mechanism. In the paper, we will also focus on the following standard (ex-post) properties of such a mechanism, to
be checked at the equilibrium $\t$ where agents truthfully report their private types:

\vspace{-2mm}
\begin{myitemize}
\item[$\rhd$] \emph{efficiency}: $A(\t)\in \Pi^*_\t$. That is, the social welfare is maximized.

\item[$\rhd$] \emph{individual rationality}: $u_{i}(A(\t),\t)\geq 0$, for each agent $i\in \A$. Hence, voluntary participation is encouraged.

\item[$\rhd$] \emph{(strong) budget-balance}: $\sum_{i\in \A}p_i(A(\t),\t)=0$. In other words, there is no transfer of money out     of or into
    the mechanism. This is in fact property \textbf{(P1)}.

\item[$\rhd$] \emph{envy-freeness}: for each pair of agents $i,j\in \A$, and for each allocation
    $\pi$ such that $\pi(i)=A(\t)(j)$, $u_{i}(A(\t),\t)\geq u_{i}(\pi,\t)$.

\item[$\rhd$] \emph{Pareto-efficiency}: there is no allocation $\pi$ such that: (1) $u_{i}(A(\t),\t)\leq u_{i}(\pi,\t)$, for each agent
    $i\in \A$, and (2) there is an agent $j\in \A$ with $u_{j}(A(\t),\t)< u_{j}(\pi,\t)$. That is, $A(\t)$ is not Pareto-dominated by any other allocation.
\end{myitemize}

Note that envy-freeness and Pareto efficiency are two direct consequences of the much stronger property \textbf{(P2)}, which can now be formalized as follows:

\vspace{-1mm}
\begin{myitemize}
\item[$\rhd$] \emph{fairness}: $u_{i}(A(\t),\t)\geq u_{i}(\pi,\t)$, for each agent $i\in \A$ and for each allocation $\pi$.
\end{myitemize}

The two remaining properties \textbf{(P3)} and \textbf{(P5)} from the list in Section~\ref{sec:motivation} (``implementability'' and ``no
punishment'', respectively) will be considered as well. Note that these properties must hold in general, i.e, not only at the equilibrium. In
fact, in the more general setting of allocation problems, they can be respectively formalized as follows:

\vspace{-1mm}
\begin{myitemize}
\item[$\rhd$] \emph{implementability}: $p_i(\pi,\w)=p_i(\pi,\w')$, for each agent $i\in \A$, for each allocation $\pi$, and for each pair
    $\w,\w'\in \D$  of vectors that that may differ only outside $\img(\pi)$, i.e., $\w_j(g)=\w'_j(g)$, for each $j\in \A$ and $g\in
    \img(\pi)$. That is, goods that are not allocated and hence not verified do not play any role in payments.

\item[$\rhd$] \emph{no punishment:}
 $p_i(\pi,\w)=p_i(\pi,(t_i,\w_{-i}))$, for each agent $i\in \A$, for each allocation $\pi$, and
    for each type vector $\w \in \D$. That is, the payment function for agent $i$ does not depend on her/his declared values, so that
    possible discrepancies with the verified type $t_i$ do not affect $i$'s payment according to the given allocation.
    In other words, we may think of payments being always computed under the presumption of innocence, where incorrect declared values do not
necessarily mean manipulation attempts by the agents.
\end{myitemize}

\section{Mechanisms with Verification for Allocation Problems}\label{sec:mechanism}

In this section, we introduce a mechanism with verification for allocation problems and start its analysis, by preliminary evidencing some
properties that hold over optimal allocations.

\subsection{General Properties of Allocation Problems}

We first observe that the optimization problem used to allocate goods to agents can be equivalently reformulated in such a way that at most one
good can be allocated to each agent. Intuitively, we may replace each agent $i$ by $\omega(i)$ fresh agents with the same properties as $i$.
We remark that such an equivalence is just used for the combinatorial optimization phase, i.e., without taking into account any (game
theoretic) incentive consideration.

Let $S=\tuple{\A,G,\omega}$ be an allocation problem. We denote by $S^{\bf 1}$ its {\em one-good} version $\tuple{\A^{\bf 1},G,\mathbf{1}}$,
where $\A^{\bf 1}$ is the set of agents $\bigcup_{i\in \A} \mathit{clones}(i)$ such that for each agent
$i\in \A$,  $\mathit{clones}(i)$ is a set of $\omega(i)$ fresh agents, and where $\mathbf{1}$ is the constant function mapping each 
agent to $1$. For any vector $\w$ mapping each agent in $\A$ to her/his type, let $\w^{\bf 1}$ be the type vector for agents in $\A^{\bf 1}$
such that $w_c^{\bf 1}=w_i$, for each $i\in \A$ and $c\in \mathit{clones}(i)$. Thus, in the allocation problem $S^{\bf 1}$ and considering the
type vector $\w^{\bf 1}$ each ``clone'' $c\in \mathit{clones}(i)$ can get at most one good, and it has the same valuations as agent $i$ in
$\w$.

\begin{figure}[t]
  \centering
  \includegraphics[width=0.9\textwidth]{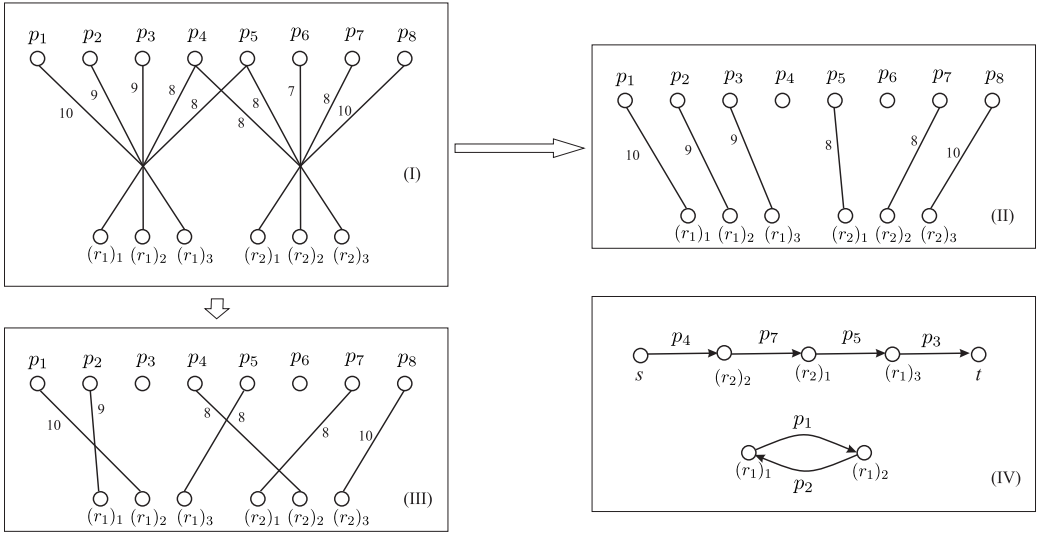}
  \caption{One-good version of the allocation problem in Example~\ref{ex:NORM}, with two allocations and their associated update graph, as defined in the proof of Theorem~\ref{thm:monotonicity}.}
  \label{fig:NORM}
\end{figure}

\begin{example}\label{ex:NORM}
Consider again the ANVUR scenario discussed in Example~\ref{ex:quattro} and illustrated in Figure~\ref{fig:esempio2}. Recall from
Example~\ref{ex:ANVURallocation} the associated allocation problem $S_R=\tuple{\R,\{p_1,...,p_8\},\mathbf{3}}$, with $\R=\{r_1,r_2\}$.
Moreover, recall that the vector $\d$ of declared types is such that, for each $i\in \{1,2\}$ and $j\in\{1,...,8\}$:
$d_{r_i}(p_j)=\score_{r_i}(p_j)$, if $p_j\in \products(r_i)$, and $d_{r_i}(p_j)=-1$, otherwise.
Its one-good version $S_R^{\bf 1}$ is shown Figure~\ref{fig:NORM}(I), where the set of agents is
$\{(r_1)_1,(r_1)_2,(r_1)_3,(r_2)_1,(r_2)_2,(r_2)_3\}$ and where we have that $d_{(r_1)_h}^{\bf 1}=d_{r_1}$ (resp., $d_{(r_2)_h}^{\bf
1}=d_{r_2}$), for each $h\in \{1,2,3\}$. In this graphical representation, crossing lines represent the edges of the bipartite cliques
connecting the two groups of virtual researchers with their products, while, as usual, edges with negative scores (valuations for products
owned only by other authors) are omitted.  \hfill $\lhd$
\end{example}

For any set  $\C\subseteq \A$ of agents and set $G'\subseteq G$ of goods, the tuple $\tuple{\C,G',\omega}$ is the restriction of
$\tuple{\A,G,\omega}$ where only agents in $\C$ and goods in $G'$ are considered.\footnote{Note the little abuse of notation: the function
$\omega$ in $\tuple{\C,G',\omega}$ should be in fact its restriction over $\C$. However, to keep the notation simple, we just write $\omega$,
as no confusion may arise. Similarly, any type vector $\w$ for $\A$ will be transparently considered as a type vector for any subset of agents
$\C\subseteq \A$---we just get rid of the unused components associated with agents in $\A\setminus \C$.}
Consider such a scenario $\tuple{\C,G',\onC{\omega}}$ and let $\pi_\C^{\bf 1}$ be an allocation for its one-good version
$\tuple{\C,G',\onC{\omega}}^{\bf 1}$. Note that $\pi_\C^{\bf 1}$ induces the allocation $\pi_\C(i)=\bigcup_{c\in\mathit{clones}(i)}\pi_\C^{\bf
1}(c)$ for $\tuple{\C,G',\onC{\omega}}$, denoted by $\omeganorm(\pi_\C^{\bf 1})$. By construction, $\val(\pi_\C,\onC{\w})=\val(\pi_\C^{\bf
1},\onC{\w}^{\bf 1})$.  Conversely, any allocation $\bar\pi_\C$ for $\tuple{\C,G',\onC{\omega}}$ is associated with the non-empty set
$\norm(\bar\pi_\C)$  of all those  allocations $\bar\pi_\C^{\bf 1}$ such that $\bar\pi_\C=\omeganorm(\bar\pi_\C^{\bf 1})$, also called the
one-good forms of  $\bar\pi_\C$.

\begin{fact}\label{lemma:norm}
Let $\mathcal{S}$ be an allocation problem, let $\mathcal{S}^{\bf 1}$ be its one-good version, and let $\pi_\C^{\bf 1}$ be an allocation for
$\mathcal{S}^{\bf 1}$. Then, $\pi_\C^{\bf 1}$ is an optimal allocation for $\mathcal{S}^{\bf 1}$ w.r.t.~$\onC{\w}^{\bf 1}$ if, and only if,
$\omeganorm(\pi_\C^{\bf 1})$ is an optimal allocation for $\mathcal{S}$ w.r.t.~$\onC{\w}$.
\end{fact}

\begin{example}
Consider again the setting of Example~\ref{ex:NORM}, and the optimal allocation $\submitted^*_{\tt o}$ shown in Figure~\ref{fig:esempio2}(II). Then, it
is immediate to check that the allocation depicted in Figure~\ref{fig:NORM}(II) is indeed an associated one-good form allocation, which is
actually an optimal allocation for $\tuple{\R,\{p_1,...,p_8\},\mathbf{3}}^{\bf 1}$ w.r.t.~$\d^{\bf 1}$, by Fact~\ref{lemma:norm}.  \hfill
$\lhd$
\end{example}

We are now in the position of stating a property that holds on any optimal allocation~$\pi$. The property is in fact of interest of its own,
i.e., independently of its application to the design of a mechanism with verification. In words, it tells us that, whenever we are interested
in allocating goods to any subset of agents, we may safely consider only goods in $\img(\pi)$, rather than the whole set $G$.  In our case, it
is a basic technical ingredient for showing a number of key properties because, intuitively, it allows us to get rid of alternative (optimal)
allocations, possibly based on non-evaluated goods in $G\setminus \img(\pi)$.

\begin{theorem}\label{thm:monotonicity}
Let $\pi$ be an optimal allocation for $\tuple{\A,G,\omega}$ w.r.t.~$\w$, and let $G_\pi= \img(\pi)\subseteq G$ be the set of goods allocated
according to $\pi$. Then, for each set of agents $\C\subseteq \A$, every optimal allocation for $\tuple{\C,G_\pi,\onC{\omega}}$
w.r.t.~$\onC{\w}$  is an optimal allocation for $\tuple{\C,G,\onC{\omega}}$ w.r.t.~$\onC{\w}$.
\end{theorem}
\begin{proof}
Let $\C\subseteq \A$ be any set of agents,  and let $\eta_\C$ be any optimal allocation for $\tuple{\C,G_\pi,\onC{\omega}}$ w.r.t.~$\onC{\w}$,
where $G_\pi=\img(\pi)$. We show that $\eta_\C$ is an optimal allocation for the unrestricted problem $\tuple{\C,G,\onC{\omega}}$ w.r.t.
$\onC{\w}$, too.

To this end, consider any optimal allocation $\lambda_\C$ for the problem $\tuple{\C,G,\onC{\omega}}$ where all goods in $G$ are available to
the agents in $\C$.  We next prove that $\val(\eta_\C,\onC{\w})=\val(\lambda_\C,\onC{\w})$. This clearly follows from the optimality of
$\eta_\C$ if $\img(\lambda_\C)\subseteq G_\pi$ holds.  Therefore, to be strictly better than $\eta_\C$, a function must allocate some good in
$G\setminus G_\pi$.
Assume thus by contradiction that $\val(\eta_\C,\onC{\w})<\val(\lambda_\C,\onC{\w})$, and hence  $\img(\lambda_\C)\not\subseteq G_\pi$, which
entails that $G_\pi \subset G$. Consider two allocations $\eta_\C^{\bf 1}\in\norm(\eta_\C)$ and $\lambda_\C^{\bf 1}\in\norm(\lambda_\C)$, and
observe first that: $ \val(\eta_\C^{\bf 1},\onC{\w}^{\bf 1})=\val(\eta_\C,\onC{\w})<\val(\lambda_\C,\onC{\w})=\val(\lambda_\C^{\bf
1},\onC{\w}^{\bf 1}).$

Let $S$ be the set of agents whose good-assignment are the same according to these allocations, i.e., $S = \{ c \in \C^{\bf 1} \mid
\lambda_\C^{\bf 1}(c)=\eta_\C^{\bf 1}(c)\}$.
Then, define $\Delta(\eta_\C^{\bf 1},\lambda_\C^{\bf 1})=(\C^{\bf 1}\setminus S\cup\{s,t\},E)$ to be the directed graph, called \emph{update}
graph for $\eta_\C^{\bf 1}$ w.r.t.~$\lambda_\C^{\bf 1}$, whose nodes are the agents in $\C^{\bf 1}$ that change their goods in the two
allocations plus two distinguished nodes $s$ and $t$, and whose edges in $E$ are defined as follows:

\vspace{-2mm}
\begin{myitemize}
  \item There is an edge from agent $c$ to agent $c'$ if $\lambda_\C^{\bf 1}(c')=\eta_\C^{\bf 1}(c)\neq \emptyset$;
  \vspace{-2mm}\item There is an edge from $s$ to agent $c'$ if 
  there is no agent $c$ such that $\lambda_\C^{\bf 1}(c')=\eta_\C^{\bf
      1}(c)\neq\emptyset$;
  \vspace{-2mm}\item There is an edge from agent $c$ to $t$ if 
  there is no agent $c'$  such that $\lambda_\C^{\bf 1}(c')=\eta_\C^{\bf 1}(c)\neq \emptyset$;
  \vspace{-2mm}\item No further edges are in $E$.
\end{myitemize}\vspace{-2mm}

\noindent  For an example construction, consider Figure~\ref{fig:NORM}(IV) showing the update graph for the allocation shown in
Figure~\ref{fig:NORM}(II) w.r.t.~the allocation shown in Figure~\ref{fig:NORM}(III).

\medskip

As each agent gets at most one good in $\eta_\C^{\bf 1}$ and $\lambda_\C^{\bf 1}$, each node in $\Delta(\eta_\C^{\bf 1},\lambda_\C^{\bf 1})$
but $s$ and $t$ has exactly one incoming edge and one outgoing edge. Moreover, by construction, $s$ has no incoming edge, and $t$ has no
outgoing edge. Thus, the update graph consists of a number of paths from $s$ to $t$ and a number of cycles, all of them being disjoint from
each other.

Let $\{\tau_1,...,\tau_h\}$ be the set of all possible paths from $s$ to $t$ or cycles in $\Delta(\eta_\C^{\bf 1},\lambda_\C^{\bf 1})$, and for
a path or a cycle $\tau_i=\alpha_1,...,\alpha_m$, let $\nodes(\tau_i)$ be the set $\{\alpha_1,...,\alpha_m\}\setminus\{s,t\}$.
In addition, let us fix the following notation: For any function $\pi':A'\mapsto G'$, let $\pi'[A']$ denote the restriction of $\pi'$ over
$A'\subseteq A$. Moreover, for the functions $\pi_1:A_1\mapsto G_1$ and $\pi_2:A_2\mapsto G_2$ with $A_1\cap A_2=\emptyset$, let
$\pi_1\biguplus\pi_2:A_1\cup A_2\mapsto G_1\cup G_2$ be such that $(\pi_1\biguplus\pi_2)[A_1]=\pi_1$ and $(\pi_1\biguplus\pi_2)[A_2]=\pi_2$.

\smallskip

By the construction of the update graph, note that $\lambda_\C^{\bf 1}$ can be expressed in terms of the disjoint paths/cycles
$\tau_1,...,\tau_h$ by the following expression:
 $$\eta_\C^{\bf 1}[\C^{\bf 1}\setminus\bigcup_{i=1}^h \nodes(\tau_i)]\ \biguplus_{i=1}^h \lambda_\C^{\bf
1}[\nodes(\tau_i)].$$ Because $\val(\eta_\C^{\bf 1},\onC{\w}^{\bf 1})<\val(\lambda_\C^{\bf 1},\onC{\w}^{\bf 1})$, there must exists a set of
agents $\nodes(\tau_k)$, associated with some disjoint path/cycle  $\tau_k$, with $1\leq k\leq h$, such that the value of the goods allocated
to these agents according to $\lambda_\C^{\bf 1}$ is greater than the corresponding value for the same agents obtained after $\eta_\C^{\bf 1}$.
That is, the function $\pi_{\tau_k}=\eta_\C^{\bf 1}[\C^{\bf 1}\setminus \nodes(\tau_{k})]\biguplus \lambda_\C^{\bf 1}[\nodes(\tau_k)]$, which
is an allocation for $\tuple{\C,G,\onC{\omega}}^{\bf 1}$, is such that $\val(\pi_{\tau_k},\onC{\w}^{\bf 1})>\val(\eta_\C^{\bf 1},\onC{\w}^{\bf
1})$.
Note that if $\tau_k$ were a cycle or a path of the form $s,\alpha_2,\dots,\alpha_{m-1},t$ such that $\lambda_\C^{\bf 1}(\alpha_{2})\subseteq
G_\pi$,
then $\img(\pi_{\tau_k})\subseteq G_\pi$ would hold. Indeed, only the first node in a path, such as $\alpha_2$, may be such that
$\lambda_\C^{\bf 1}(\alpha_{2})\setminus G_\pi\neq\emptyset$. However, as observed above, this is impossible because
$\val(\pi_{\tau_k},\onC{\w}^{\bf 1})>\val(\eta_\C^{\bf 1},\onC{\w}^{\bf 1})$ would contradict the optimality of $\eta_\C^{\bf 1}$, and hence
the optimality of $\eta_\C$, by Fact~\ref{lemma:norm}.

Therefore, we can conclude that $\tau_k$ is a path of the form $s,\alpha_2,\dots,\alpha_{m-1},t$ with
$\pi_{\tau_k}(\alpha_2)=\lambda(\alpha_2)=\{g'\}\in G\setminus G_\pi$. That is, the allocation $\pi_{\tau_k}$ (over the agents in $\C^{\bf 1}$)
is such that
 $\img(\pi_{\tau_k})=\{g'\}\cup \img(\eta_\C^{\bf 1})\setminus \eta_\C^{\bf 1}(\alpha_{m-1})$.
 In particular, observe that
  $\eta_\C^{\bf 1}(\alpha_{m-1})\subseteq G_\pi\setminus\img(\pi_{\tau_k})$.

\smallskip

Let us now come back to the optimal allocation $\pi$ for $\tuple{\A,G,\omega}$ w.r.t.~$\w$, and let $\pi^{\bf 1}$ be an (optimal) allocation in
$\norm(\pi)$. Let $A\subseteq \C^{\bf 1}$ be a set of agents with $\alpha_2\in A$ such that the set of goods $\bigcup_{c\in A}\pi^{\bf 1}(c)$
allocated to these agents according to $\pi^{\bf 1}$ is equal to $\bigcup_{c\in A}\pi_{\tau_k}(c)\setminus\{g'\}\cup G'',$
where $G''\subseteq G_\pi\setminus\img(\pi_{\tau_k})$ and $|G''|\leq 1$. Note that a set $A$ having this property in fact exists: just start
with $\{\alpha_2\}$ and then add agents from $\nodes(\tau_k)$ until some $c$ is found with $\pi^{\bf 1}(c)\subseteq
G_\pi\setminus\img(\pi_{\tau_k})$.

Consider then $\bar\pi=\pi_{\tau_k}[A]\biguplus \pi^{\bf 1}[\A^{\bf 1}\setminus A]$ and note that $\bar\pi$ is indeed an allocation for
$\tuple{\A,G,{\omega}}^{\bf 1}$, because the construction of the set $A$ guarantees that no good allocated according to $\pi_{\tau_k}[A]$ can
be allocated by $\pi^{\bf 1}$ to agents in $[\A^{\bf 1}\setminus A]$, and vice-versa.
Since $\pi^{\bf 1}$ is an optimal allocation for $\tuple{\A,G,{\omega}}^{\bf 1}$ w.r.t.~$\w^{\bf 1}$,
 $\val(\pi^{\bf 1},\onC{\w}^{\bf 1})\geq \val(\bar\pi,\onC{\w}^{\bf 1})$ holds.
Thus, by construction of $\bar\pi$, we get $\val(\pi^{\bf 1}[{A}],\onC{\w}^{\bf 1})\geq \val(\bar\pi[{A}],\onC{\w}^{\bf
1})=\val(\pi_{\tau_k}[A],\onC{\w}^{\bf 1})$.

Finally, let $\bar\pi'_\C=\pi_{\tau_k}[\C^{\bf 1}\setminus A]\biguplus \pi^{\bf 1}[A]$ and note that $\bar\pi_\C'$ is an allocation for
$\tuple{\C,G,\onC{\omega},\onC{\w}}^{\bf 1}$.
Moreover, observe that $\val(\bar\pi_\C',\onC{\w}^{\bf 1})\geq \val(\pi_{\tau_k},\onC{\w}^{\bf 1})>\val(\eta_\C^{\bf 1})$ and
$\img(\bar\pi_\C')\subseteq G_\pi$. For this latter, just recall that $\alpha_2\in A$ is the only agent in $\C^{\bf 1}$ such that
$\pi_{\tau_k}(\alpha_2)\setminus G_\pi\neq \emptyset$. Again, this entails that $\eta_\C^{\bf 1}$ is not optimal w.r.t.~$\onC{\w}^{\bf 1}$ and
hence by Fact~\ref{lemma:norm} $\eta_\C$ is also not optimal for $\tuple{\C,G_\pi,\onC{\omega}}$ w.r.t.~$\onC{\w}$. Contradiction.~\hfill~
\end{proof}

The result immediately entails the following two corollaries.

\begin{corollary}\label{cor:uno}
For each optimal allocation $\pi$ for $\tuple{\A,G,\omega}$ w.r.t.~$\w$ and for each set $\C\subseteq \A$ of agents,
$\opt(\tuple{\C,\img(\pi),\onC{\omega}},\onC{\w})=\opt(\tuple{\C,G,\onC{\omega}},\onC{\w})$.
\end{corollary}

\begin{corollary}\label{cor:due}
Let $\pi$ be an optimal allocation for $\tuple{\A,G,\omega}$ w.r.t.~$\w$, and let $\pi'$ be any allocation for $\tuple{\A,G,\omega}$, hence
with $\val(\pi,\w)\geq \val(\pi',\w)$.
Then, for each set $\C\subseteq \A$ of agents, $\opt(\tuple{\C,\img(\pi),\onC{\omega}},\onC{\w})\geq \opt(\tuple{\C,\img(
\pi'),\onC{\omega}},\onC{\w})$.
\end{corollary}

\begin{figure}[t]
\centering  \fbox{
\parbox{0.98\textwidth}
{
\begin{tabular}{ll}
\textbf{Input}:     &\ An allocation $\pi$ for $\tuple{\A,G,\omega}$, and a vector $\w\in \D$;  \\
\textbf{Assumption}: &\ A verifier $\v$ is available. Let $\v(\pi)=(v_1,...,v_n)$;
\end{tabular}\\
\vspace{-0mm}\hrule\vspace{1mm}
\begin{tabular}{ll}
   1.\ \ \ & Let $\mymath{C}$ denote the set of all possible subsets of $\A$;    \\
   2.\ \ \ & For each set $\C\in \mymath{C}$,\\
   3.\ \ \ & $\lfloor$\ \ \ \ Compute an optimal allocation $\pi_\C$ for $\tuple{\C,\img(\pi),\onC{\omega}}$ w.r.t.~$\onC{\w}$; \\
   4.\ \ \ & For each agent $i\in \A$,\\
   5.\ \ \ & $|$\ \ \ \ For each set $\C\in \mymath{C}$,\\
   6.\ \ \ & $|$\ \ \ \ $|$\ \ \ \ Let $\Delta^1_{\C,i}(\pi,\w):=\val(\pi_\C,\onC{(v_i,\w_{-i})})$; \ \ \ \ \  (=$v_i(\pi_{\C})+\sum_{j\in \C\setminus\{i\}}w_j(\pi_{\C})$); \\
   7.\ \ \ & $|$\ \ \ \ \hspace{-0.2mm}$\lfloor$\ \ \ \ Let $\Delta^2_{\C,i}(\pi,\w):=\val(\pi_{\C\setminus\{i\}},\onC{\w})$; \ \ \ \ \ (=$\sum_{j\in \C\setminus\{i\}}w_j(\pi_{\C\setminus\{i\}})$); \\
   8.\ \ \ & $|$\ \ \ \ Let $\xi_i(\pi,\w):=\sum_{\C \in \mymath{C}}\frac{(|\A|-|\C|)!(|\C|-1)!}{|\A|!}(\Delta^1_{\C,i}(\pi,\w)-\Delta^2_{\C,i}(\pi,\w))$; \\
   9.\ \ \ & \hspace{-0.2mm}$\lfloor$\ \ \ \ Define $p_i^\xi(\pi,\w):=\xi_i(\pi,\w)-v_i(\pi)$;    \\
\end{tabular}
} }\vspace{-1mm}
  \caption{Payment rule $\p^\xi$.} \label{fig:M1}
\end{figure}

\subsection{The Design of a Truthful Mechanism}

With the above notation and results in place, we can now discuss the payment rule $\p^\xi$ that is illustrated in Figure~\ref{fig:M1}:
We are given an allocation $\pi$ that selects some goods $\img(\pi)\subseteq G$ for the agents in $\A$, plus a vector $\w\in \D$. Moreover, we
assume the existence of a verifier computing the vector $\v(\pi)=(v_1,...,v_n)$.

In the first three steps, the payment rule associates an optimal allocation $\pi_\C$ for $\tuple{\C,\img(\pi),\onC{\omega}}$ w.r.t.~$\onC{\w}$
with each set $\C\in \mymath{C}$ of agents, where $\mymath{C}$ is the powerset of $\A$, i.e., the set of all possible subsets of agents. Then,
for each agent $i\in \A$ and for each set $\C\in \mymath{C}$, we define two terms, namely $\Delta^1_{\C,i}(\pi,\w)$ and
$\Delta^2_{\C,i}(\pi,\w)$, which evaluate the allocations $\pi_\C$ and $\pi_{\C\setminus\{i\}}$ under the assumption that agent types are
$\onC{(v_i,\w_{-i})}$ and $\onCi{\w}$, respectively.
These terms will play a role in the definition of the value $\xi_i(\pi,\w)$ at step~8.
In particular, we observe that the definition of this value is reminiscent of the definition of the \emph{Shapley value}, as it considers the
marginal contribution of each possible set $\C$ summed up in a weighted manner w.r.t.~its size.\footnote{The reader that is not
familiar with this solution concept is referred to the introductory part of Section~\ref{sec:coalitional}, where the Shapley value is formally defined.}
Finally, the payment $p_i^\xi(\pi,\w)$ is defined at step~9 as the difference between $\xi_i(\pi,\w)$ and $v_i(\pi)$.

Note that the idea underlying the definition of $\p^\xi$ is that, after verification is performed, the utility function will precisely coincide
with the ``bonus'' $\xi_i(\pi,\w)$, hence sharing the spirit\footnote{In fact, the peculiar form of $\xi_i(\pi,\w)$ does not fit the general
schema by~\citeA{Nisan2001}.} of the approach by~\citeA{Nisan2001}.
Indeed, the following is immediate.

\begin{lemma}\label{lemma:utility}
For each allocation $\pi$ for $\tuple{\A,G,\omega}$, for each vector $\w\in \D$, and for each agent $i\in \A$, it holds that
$u_{i}(\pi,\w)=\xi_i(\pi,\w)$.
\end{lemma}

By exploiting this characterization, we can now show the first crucial result on the payment rule $\p^\xi$, i.e., that the mechanism
$(A,\p^\xi)$---where $A$ is any arbitrary optimal allocation algorithm---is truthful.

\begin{theorem}[\textbf{truthfulness}]\label{thm:truth}
Let $A$ be any optimal allocation algorithm. Then, the mechanism with verification $(A,\p^\xi)$ is truthful.
\end{theorem}
\begin{proof}
We have to show that, for each agent $i\in \A$, and reported type vector $\d$, the following holds:
$u_{i}(A(\mydot{t_i}{\d_{-i}}),(\mydot{t_i}{\d_{-i}}))\geq u_{i}(A(\d),\d)$; hence, by Lemma~\ref{lemma:utility}, that
$\xi_i(A(\mydot{t_i}{\d_{-i}}),(t_i,\d_{-i}))\geq \xi_i(A(\d),\d)$.

Consider the construction reported in Figure~\ref{fig:M1} for the two cases of $\w=\d$ and $\w=(\mydot{t_i}{\d_{-i}})$, and let $\pi=A(\d)$ and
$\pi'=A(\mydot{t_i}{\d_{-i}})$ be the corresponding allocations (optimal w.r.t.~$\d$ and $(t_i,\d_{-i})$, respectively) received as input by
the payment rule in Figure~\ref{fig:M1}. For any set $\C\in \mymath{C}$ of agents, let $\pi_C$ (resp., $\pi'_\C$) be the allocation computed at
step~3. We show that the following two properties hold, for each set of agents $\C\in \mymath{C}$:
\begin{myitemize}
  \item[(A)] $\Delta^1_{\C,i}(\pi',(\mydot{t_i}{\d_{-i}}))\geq  \Delta^1_{\C,i}(\pi,(\mydot{b_i}{\d_{-i}}))$, and
  \item[(B)] $\Delta^2_{\C,i}(\pi',(\mydot{t_i}{\d_{-i}})) = \Delta^2_{\C,i}(\pi,(\mydot{b_i}{\d_{-i}}))$.
\end{myitemize}

In order to prove (A), observe that by step~6, $\Delta^1_{\C,i}(\pi',(\mydot{t_i}{\d_{-i}}))=v_i(\pi_{\C}')+\sum_{j\in
\C\setminus\{i\}}d_j(\pi_{\C}')=t_i(\pi_{\C}')+\sum_{j\in \C\setminus\{i\}}d_j(\pi_{\C}')$, because in this case $i$ reports the actual private
type $t_i$, which is equal to the verified one.
Now, recall that $\pi'_\C$ is an optimal allocation for $\tuple{\C,\img(\pi'),\onC{\omega}}$ w.r.t.~$\onC{(\mydot{t_i}{\d_{-i}})}$, and
$\pi'=A(\mydot{t_i}{\d_{-i}})$ is an optimal allocation for
$\tuple{\A,G,\omega}$ w.r.t.~$(\mydot{t_i}{\d_{-i}})$. Thus, by Corollary~\ref{cor:uno}, 
\begin{equation}\label{eq:d2}
t_i(\pi_{\C}')+\sum_{j\in \C\setminus\{i\}}d_j(\pi_{\C}')=
 \val(\pi_{\C}',\onC{(\mydot{t_i}{\d_{-i}})}) = \opt(\tuple{\C,G,\onC{\omega}},\onC{(\mydot{t_i}{\d_{-i}})}).
\end{equation}
Similarly, $\Delta^1_{\C,i}(\pi,(\mydot{b_i}{\d_{-i}}))=v_i(\pi_{\C})+\sum_{j\in \C\setminus\{i\}}d_j(\pi_{\C})$. In this case, observe that the reported type of
$i$ is $b_i$, which is different from $t_i$. Also, $\pi_\C$ is an allocation for $\tuple{\C,\img(\pi),\onC{\omega}}$, hence it allocates
goods in $\img(\pi)$, for which the true type is revealed in $\v(\pi)$. Thus, we may compute
$\Delta^1_{\C,i}(\pi,(\mydot{b_i}{\d_{-i}}))=v_i(\pi_{\C})+\sum_{j\in \C\setminus\{i\}}d_j(\pi_{\C})= t_i(\pi_{\C})+\sum_{j\in \C\setminus\{i\}}d_j(\pi_{\C})$.

In order to conclude, let us note that $\pi_\C$ is an allocation for $\tuple{\C,G,\onC{\omega}}$, though it is not necessarily optimal
w.r.t.~$\onC{(\mydot{t_i}{\d_{-i}})}$. Thus, by using Equation~\ref{eq:d2},  $t_i(\pi_{\C}')+\sum_{j\in
\C\setminus\{i\}}d_j(\pi_{\C}')=\opt(\tuple{\C,G,\onC{\omega}},\onC{(\mydot{t_i}{\d_{-i}})})\geq t_i(\pi_{\C})+\sum_{j\in
\C\setminus\{i\}}d_j(\pi_{\C})$. This shows that (A) holds.

Let us now focus on (B). By step~7, we preliminary observe that we have $\Delta^2_{\C,i}(\pi',(\mydot{t_i}{\d_{-i}}))=\sum_{j\in
\C\setminus\{i\}}d_j(\pi'_{\C\setminus\{i\}})$ and $\Delta^2_{\C,i}(\pi,(\mydot{b_i}{\d_{-i}}))=\sum_{j\in
\C\setminus\{i\}}d_j(\pi_{\C\setminus\{i\}})$. Then, recall that $\pi_{\C\setminus\{i\}}$ is an optimal allocation for
$\tuple{\C\setminus\{i\},\img(\pi),\onCi{\omega}}$ w.r.t.~$\onCi{(\mydot{t_i}{\d_{-i}})}$ and $\pi=A(\mydot{b_i}{\d_{-i}})$ is an optimal
allocation for $\tuple{\A,G,\omega}$ w.r.t.~$(\mydot{b_i}{\d_{-i}})$. Thus, by Corollary~\ref{cor:uno}, and because $i$'s evaluation is
immaterial here, we get (B) as follows: $\sum_{j\in
\C\setminus\{i\}}d_j(\pi'_{\C\setminus\{i\}})=\opt(\tuple{\C\setminus\{i\},G,\onCi{\omega}},\onCi{(\mydot{t_i}{\d_{-i}})}) =
\opt(\tuple{\C\setminus\{i\},G,\onCi{\omega}},\onCi{(\mydot{b_i}{\d_{-i}})}) = \sum_{j\in \C\setminus\{i\}}d_j(\pi_{\C\setminus\{i\}}).
$~\hfill~
\end{proof}

\medskip

In addition to truthfulness, it suddenly emerges that the payment rule is indifferent w.r.t.~deviations from the actual values (possibly,
cheats) on goods that do not occur in the allocation being selected. We have already observed that this property is relevant in our motivating
scenario, where payments should be made only with respect to certified goods (as evaluated by a third-party agency). Now, we state it in formal
terms.

\begin{theorem}[\textbf{implementability}]\label{fact:robust}
For each allocation $\pi$ for $\tuple{\A,G,\omega}$ and for each agent $i\in \A$,  it holds that $p_i^\xi(\pi,\w)=p_i^\xi(\pi,\w')$, for each
pair $\w,\w'\in \D$ of vectors that differ only outside $\img(\pi)$.
\end{theorem}
\begin{proof}
It suffices to observe that, in the algorithm in Figure~\ref{fig:M1}, optimal allocations are restricted over the set $\img(\pi)$ at step~3.
Thus, the payment rule is completely indifferent w.r.t.~agents' evaluations of the goods outside $\img(\pi)$.
\end{proof}

\medskip

Moreover, it can be noticed that payments do not depend on possible discrepancies between declared and verified values, as shown below.

\begin{theorem}[\textbf{no punishment}]
For each agent $i\in \A$, for each allocation $\pi$ for $\tuple{\A,G,\omega}$, and for each type vector $\w \in \D$, it holds that
$p_i^\xi(\pi,\w)=p_i^\xi(\pi,(t_i,\w_{-i}))$.
\end{theorem}
\begin{proof}
Recall that, in the algorithm in Figure~\ref{fig:M1}, only goods in $\img(\pi)$ are considered, and observe that the payment for agent $i$
depends only on her/his verified type, rather than on the declared one.
\end{proof}

\subsection{Further Properties of Truthful Strategies}\label{sec:fairness}

Let us now analyze some relevant properties that hold on whenever agents choose their dominant strategy of truthfully reporting their private
types. The first property is a useful characterization for the utility of the agents.

\begin{theorem}\label{thm:shapley1}
For each optimal allocation $\pi$ for $\tuple{\A,G,\omega}$ w.r.t.~$\t$, and for each agent $i\in \A$, it holds that: {\small
$$u_{i}(\pi,\t)=\sum_{\C \in
\mymath{C}}\frac{(|\A|-|\C|)!(|\C|-1)!}{|\A|!}\left(\frac{ }{ }\opt(\tuple{\C,G,\onC{\omega}},\onC{\t})-\opt(\tuple{\C\setminus\{i\},G,\onC{\omega}},\onCi{\t})\ \right).$$}
\end{theorem}
\begin{proof}
By Lemma~\ref{lemma:utility}, we know that $u_{i}(\pi,\t)=\xi_i(\pi,\t)$. Then, for each set $\C\in \mymath{C}$ of agents, and for each agent
$i\in \A$, consider the expressions $\Delta^1_{\C,i}(\pi,\t)$ and $\Delta^2_{\C,i}(\pi,\t)$ defined at step~6 and step~7, respectively, of the
mechanism in Figure~\ref{fig:M1}. Note that $\Delta^1_{\C,i}(\pi,\t)=t_i(\pi_{\C})+\sum_{j\in \C\setminus\{i\}}t_j(\pi_{\C})=\val(\pi_C)$ and
$\Delta^2_{\C,i}(\pi,\t)=\sum_{j\in \C\setminus\{i\}}t_j(\pi_{\C\setminus\{i\}})=\val(\pi_{\C\setminus\{i\}})$, where $\pi_\C$ and
$\pi_{\C\setminus\{i\}}$ are optimal allocations for $\tuple{\C,\img(\pi),\onC{\omega}}$ w.r.t.~$\onC{\t}$ and for
$\tuple{\C\setminus\{i\},\img(\pi),\onCi{\omega}}$ w.r.t.~$\onCi{\t}$, respectively.
Thus, $\Delta^1_{\C,i}(\pi,\t)=\opt(\tuple{\C,\img(\pi),\onC{\omega}},\onC{\t})$ and
$\Delta^2_{\C,i}(\pi,\t)=\opt(\tuple{\C\setminus\{i\},\img(\pi),\onC{\omega}},\onCi{\t})$. It follows that:

{\small
\begin{equation}\label{eq:img}
u_{i}(\pi,\t)=\sum_{\C\in \mymath{C}}\frac{(|\A|-|\C|)!(|\C|-1)!}{|\A|!}\left(\opt(\tuple{\C,\img(\pi),\omega},\t)-\opt(\tuple{\C\setminus\{i\},\img(\pi),\omega},\t)\right).
\end{equation}}

Recall now by Corollary~\ref{cor:uno} that, for each optimal allocation $\pi$ for $\tuple{\A,G,\pi}$ w.r.t.~$\w$ and for each set $\C\in
\mymath{C}$ of agents, $\opt(\tuple{\C,\img(\pi),\onC{\omega}},\onC{\w})=\opt(\tuple{\C,G,\onC{\omega}},\onC{\w})$. Therefore,
$\Delta^1_{\C,i}(\pi,\t)=\opt(\tuple{\C,G,\onC{\omega}},\onC{\t})$ and
$\Delta^2_{\C,i}(\pi,\t)=\opt(\tuple{\C\setminus\{i\},G,\onC{\omega}},\onCi{\t})$. By using these equalities, the result follows from
Equation~\ref{eq:img}.~\hfill~
\end{proof}

Note that in the above expression, agents' utilities are completely independent of the particular optimal allocation $\pi$. Therefore, in every
optimal allocation, every agent gets precisely the same utility.

\begin{corollary}\label{cor:independence}
Let $\pi$ and $\pi'$ be two optimal allocations for $\tuple{\A,G,\omega}$ w.r.t.~$\t$. Then, $u_{i}(\pi,\t)= u_{i}(\pi',\t)$ holds, for each $i\in \A$.
\end{corollary}

\begin{example}\label{ex:final}
Consider the allocation scenario $S_R=\tuple{\R,G,\mathbf{3}}$, where $G=\{p_1,...,p_8\}$, associated with the research structure $R$ in Example~\ref{ex:uno}.
 Assume that researchers declare their true types $\t$ in the ANVUR evaluation, and consider the optimal allocation
$\submitted^*$ shown in Figure~\ref{fig:esempio}(II). Then, we have: {
$$
\begin{array}{ll}
    u_{r_1}(\submitted^*,\t)= & \frac{1}{2}(\opt(\tuple{\{r_1,r_2\},G,\mathbf{3}},\t)-\opt(\tuple{\{r_2\},G,\mathbf{3}},\t))+\\
      & \frac{1}{2}(\opt(\tuple{\{r_1\},G,\mathbf{3}},\t)-\opt(\tuple{\{\},G,\mathbf{3}},\t))+\\
      & \frac{1}{2}(\opt(\tuple{\{r_2\},G,\mathbf{3}},\t)-\opt(\tuple{\{r_2\},G,\mathbf{3}},\t)=\\
      & \frac{1}{2}(51-26)+\frac{1}{2}(26-0)+\frac{1}{2}(26-26)=\frac{51}{2}.
\end{array}
$$}

In particular, note that $\opt(\tuple{\{r_1\},G,\mathtt{3}},\t))=26$, as we can allocate $p_1$, $p_4$, and $p_5$ to $r_1$, if
(s)he where the only researcher in the structure.

Similarly, we get $u_{r_2}(\submitted^*,\t)=\frac{51}{2}$. That is, the two researchers will share precisely one half of the total score of
their structure, based on our payment scheme.
In fact, by looking at the allocation $\submitted^*$ in Figure~\ref{fig:esempio}(II), one might na\"ively suppose that $r_2$ contributed more
than $r_1$. However, this is only due to the specific allocation considered, and not to the actual values of the products of the two authors.
For instance, $p_5$ is allocated to $r_2$ but it was produced by $r_1$, as well. Indeed, the fairness of the utility values resulting from our
payment rule suddenly appears when considering the alternative allocation $\hat{\submitted}^*$ in Figure~\ref{fig:esempio}(II), which is
symmetric w.r.t.~$\submitted^*$ and where it seems that $r_1$ contributes more than $r_2$: As a matter of fact, the two researchers are
completely interchangeable over optimal allocations, and this is correctly reflected by our payment scheme.
In particular, from Corollary~\ref{cor:independence}, the researchers are indifferent w.r.t.~the specific optimal allocation being selected,
and hence in this case they equally divide all the available score between themselves.\hfill $\lhd$
\end{example}

We conclude by pointing out a further important property of the mechanism.

\begin{corollary}[\textbf{individual-rationality}]\label{thm:efir}
Let $A$ be any optimal allocation algorithm. Then, the mechanism with verification $(A,\p^\xi)$ is individually-rational.
\end{corollary}

\begin{proof}
Notice that $\opt(\tuple{\C,G,\onC{\omega}},\onC{\t})-\opt(\tuple{\C\setminus\{i\},G,\onCi{\omega}},\onCi{\t})\geq 0$, holds
for each $\C\subseteq \A$ and agent $i\in \A$. Then, by Theorem~\ref{thm:shapley1}, $u_{i}(\pi,\t)\geq 0$ holds, for each agent $i\in
\A$.~\hfill~
\end{proof}

\section{A Coalitional Game Theory Viewpoint}\label{sec:coalitional}

A \emph{coalitional game} can be modeled as a pair $\game=\tuple{N,\varphi}$, where $N=\{1,...,n\}$ is a finite set of agents, and $\varphi$
is a function associating with each \emph{coalition} $C\subseteq N$ a real-value $\varphi(C)\in \mathbb{R}$, with $\varphi(\{\})=0$, which is
meant to encode the worth that agents in $C$ obtain by collaborating with each other.
The function $\varphi$ is \emph{supermodular} (resp., \emph{submodular}) if $\varphi(R\cup T)+\varphi(R\cap T)\geq \varphi(R)+\varphi(T)$
(resp., $\varphi(R\cup T)+\varphi(R\cap T)\leq \varphi(R)+\varphi(T)$) holds, for each pair of coalitions $R,T\subseteq N$.

A fundamental problem for coalitional games is to single out the most desirable outcomes, usually called \emph{solution concepts}, in terms of
appropriate notions of worth distributions, i.e., of vectors of payoffs ${\bf x}=(x_1,...,x_{n})\in\mathbb{R}^{n}$ such that $\sum_{i\in N}
x_i=\varphi(N)$. This question was studied in economics and game theory with the aim of providing arguments and counterarguments about why such
proposals are reasonable mathematical renderings of the intuitive concepts of fairness and stability. For further background on coalitional
games, the reader is referred to, e.g.,~\cite{Osborne1994}.

Here, we consider the \emph{Shapley value} of $\game=\tuple{N,\varphi}$, which is a well-known solution concept such that:

\vspace{-2mm}
$$\sv_i(\game)=\sum_{C \subseteq N}\frac{(|N|-|C|)!(|C|-1)!}{|N|!}(\varphi(C)-\varphi(C\setminus\{i\}))\mbox{, for each }i\in N.$$

Indeed, we shall show that the mechanism defined in Section~\ref{sec:mechanism} has a nice interpretation in terms of the Shapley value of some
suitable-defined coalitional games. The correspondence will be exploited to prove further properties of our mechanism.

\subsection{The Shapley Value of Allocation Games}

We consider two coalitional games defined on top of an allocation problem.

\begin{defn}\label{def:games}{
Given the tuple $S=\tuple{\A,G,\omega}$ and a vector $\w$ of agent types, we define $\game_{S,\w}^{\tmarg}=\tuple{\A,\marg_{S,\w}}$ and
$\game_{S,\w}^{\tbest}=\tuple{\A,\best_{S,\w}}$ as the coalitional games such, that for each set $\C\subseteq \A$ of agents,
\begin{myitemize}
\item[$\bullet$] $\marg_{\tuple{\A,G,\omega},\w}(\C)=\opt(S,\w)-\opt(\tuple{\A\setminus\C,G,\omega},\w)$; and,

\item[$\bullet$] $\best_{S,\w}(\C)=\opt(\tuple{\C,G,\omega},\w)$. \hfill $\Box$
\end{myitemize}}
\end{defn}

Note that $\marg_{S,\w}(\C)$ is the \emph{marginal contribution} of $\C$ to $\opt(S,\w)$. Instead, $\best_{S,\w}(\C)$ is the \emph{best
contribution} of $\C$ (w.r.t.~$\w$), computed assuming that agents in $\C$ were the only agents in the allocation problem.
In particular, the game $\game_{S,\w}^{\tbest}$ has already been considered by~\citeA{Moulin1992}, precisely in the setting of fair division
for allocation problems. There, it is shown that the cost function associated with $\game_{S,\w}^{\tbest}$ is \emph{submodular}.

\begin{proposition}[\citeA{Moulin1992}]\label{prop:submodular}
The function $\best_{S,\w}$ is submodular.
\end{proposition}

Then, the (dual) analogous for $\marg_{S,\w}$ can be shown easily.

\begin{theorem}\label{thm:supermodular}
The function $\marg_{S,\w}$ is supermodular.
\end{theorem}

\begin{proof}
Let $S=\tuple{\A,G,\omega}$ be the given structure, and $\w$ be a vector of types.
The result just follows by noticing that
$\marg_{S,\w}(\C)=\opt(\tuple{\A,G,\omega},\w)-\opt(\tuple{\A\setminus\C,G,\omega},\w)=\opt(\tuple{\A,G,\omega},\w)-\best_{S,\w}(\A\setminus\C)$,
for each set of agents $\C\subseteq \A$. That is, $\best_{S,\w}(\C)=\opt(\tuple{\A,G,\omega},\w)-\marg_{S,\w}(\A\setminus\C)$. Thus, if
$\best_{S,\w}(R\cup T)+\best_{S,\w}(R\cap T)\leq \best_{S,\w}(R)+\best_{S,\w}(T)$ holds $\forall R,T\subseteq N$, we have that
$\marg_{S,\w}(\A\setminus (R\cup T))+\marg_{S,\w}(\A\setminus (R\cap T))\geq \marg_{S,\w}(\A\setminus R)+\marg_{S,\w}(\A\setminus T)$ holds as
well, $\forall R,T\subseteq N$.
Eventually, by letting $R'=\A\setminus R$ and $T'=\A\setminus T$, we get $\marg_{S,\w}(R'\cap T')+\marg_{S,\w}(R'\cup T'))\geq
\marg_{S,\w}(R')+\marg_{S,\w}(T')$, for each $\forall R',T'\subseteq N$. That is, $\marg_{S,\w}$ is supermodular.~\hfill~
\end{proof}

As a second relevant property, we next observe that the payment rules in Section~\ref{sec:mechanism} coincide, at the equilibrium $\t$ where
agents truthfully report their types, with the Shapley value of the game $\game_{S,\t}^{\tbest}$ associated with $S$. The result follows by
comparing the utility function as in Theorem~\ref{thm:shapley1} with the expression for the Shapley value of the coalitional game
$\game_{S,\t}^{\tbest}$.
Moreover, we show that the same result can be established for the ``dual'' game $\game_{S,\t}^{\tmarg}$, so that the Shapley values of the two
games are identical---for similar correspondences between Shapley values of different games, see also the works by~\citeA{Francois2003} and
\citeA{Kalai1983}.

\begin{theorem}\label{thm:shapley3}
For each optimal allocation $\pi$ for $S=\tuple{\A,G,\omega}$ w.r.t.~$\t$, and for each agent $i\in \A$, it holds that
$u_{i}(\pi,\t)=\xi_i(\pi,\t)=\sv_i(\game_{S,\t}^{\tbest})=\sv_i(\game_{S,\t}^{\tmarg})$.
\end{theorem}

\begin{proof}
By comparing the utility function as in Theorem~\ref{thm:shapley1} with the expression for the Shapley value of the coalitional game
$\game_{S,\t}^{\tbest}$ associating with each coalition $\C$ of agents the worth $\opt(\tuple{\C,G,\omega},\t)$, we immediately get that, for
each optimal allocation $\pi$ for $S$ w.r.t.~$\t$, and for each agent $i\in \A$, it holds that
$u_{i}(\pi,\t)=\xi_i(\pi,\t)=\sv_i(\game_{S,\t}^{\tbest})$.

In order to conclude the proof, we show that for each agent $i\in \A$, $\sv_i(\game_{S,\t}^{\tmarg})=\sv_i(\game_{S,\t}^{\tbest})$ holds. To
this end, first note that these Shapley values can be written as follows:
\begin{myitemize}
  \item $\sv_i(\game_{S,\t}^{\tmarg})=\sum_{\C \subseteq \A, i\in \C}\frac{(|\A|-|\C|)!(|\C|-1)!}{|\A|!} T'_\C$, and

  \item $\sv_i(\game_{S,\t}^{\tbest})=\sum_{\C \subseteq \A, i\in \C}\frac{(|\A|-|\C|)!(|\C|-1)!}{|\A|!} T_\C$,
\end{myitemize}
\noindent where $T'_\C=\marg_{S,\t}(\C)-\marg_{S,\t}(\C\setminus\{i\})$ and
      $T_\C=\opt(\tuple{\C,G,\omega},\t)-\opt(\tuple{\C\setminus\{i\},G,\omega},\t).$

Then, we claim that:
\begin{myitemize}
\item[(1)] for each set $\C\subseteq \A$ of agents with $i\in \C$, the set $\bar \C=(\A\setminus\C)\cup\{i\}$ is such that
    $T'_\C=T_{\bar \C}$, and

\item[(2)] for each set $\bar \C\subseteq \A$ of agents with $i\in \bar \C$, the set $\C=(\A\setminus \bar \C)\cup\{i\}$ is such that
    $T'_\C=T_{\bar \C}$.
\end{myitemize}

(1) Let $\C\subseteq \A$ such that $i\in \C$, and observe that
$T_\C'=\marg_{S,\t}(\C)-\marg_{S,\t}(\C\setminus\{i\})=(\opt(S,\t)-\opt(\tuple{\A\setminus\C,G,\omega},\t))-
(\opt(S,\t)-\opt(\tuple{\A\setminus(\C\setminus\{i\}),G,\omega},\t))=\opt(\tuple{\A\setminus(\C\setminus\{i\}),G,\omega},\t)-\opt(\tuple{\A\setminus\C,G,\omega},\t)=\opt(\tuple{(\A\setminus\C)\cup\{i\}),G,\omega},\t)-\opt(\tuple{\A\setminus\C,G,\omega},\t)$.
Thus, let $\bar \C=(\A\setminus\C)\cup\{i\}$, and note that $T_\C'=T_{\bar \C}$.

(2) Let $\bar \C\subseteq \A$ such that $i\in \bar \C$, and observe that
$T_{\bar \C}=\opt(\tuple{\bar \C,G,\omega},\t)-\opt(\tuple{\bar \C\setminus\{i\},G,\omega},\t)=(\opt(S,\t)-\opt(\tuple{\bar
\C\setminus\{i\},G,\omega},\t))-(\opt(S,\t)-\opt(\tuple{\bar \C,G,\omega},\t))=\marg_{S,\t}((\A\setminus \bar
\C)\cup\{i\})-\marg_{S,\t}(\A\setminus \bar \C)$. Thus, let $\C=(\A\setminus \bar \C)\cup\{i\}$ and note that $T_\C'=T_{\bar \C}$.

\smallskip

As (1) and (2) hold, and given the two expressions for $\sv_i(\game_{S,\t}^{\tmarg})$ and $\sv_i(\game_{S,\t}^{\tbest})$, we conclude that the
two values coincide.~\hfill~
\end{proof}

\subsection{Fairness and Budget-Balancedness}

Now that we have established a precise correspondence between our mechanism and the Shapley value of its associated allocation games, we can
show further desirable properties of~$\p^\xi$.
In fact, we exploit the following well-known properties~\cite<see, e.g., >{Osborne1994,Young1985} of the Shapley value of any game
$\game=\tuple{N,\varphi}$:

  \begin{myitemize}
  \item[(I)] $\sum_{i\in N} \sv_i(\game)=\varphi(N)$;

  \vspace{1mm}\item[(II)] If $\varphi$ is \emph{supermodular} (resp., \emph{submodular}), then $\sum_{i\in C} \sv_i(\game)\geq \varphi(C)$
      (resp., $\sum_{i\in C} \sv_i(\game)\leq \varphi(C)$).

  \vspace{1mm}\item[(III)] If $\game'=\tuple{N,\varphi'}$ is a game such that $\varphi'(C)\geq \varphi(C)$, for each $C\subseteq N$, then
  $\sv_i(\game')\geq \sv_i(\game)$, for each agent $i\in \A$.
\end{myitemize}

Our first result is to show that $\best_{S,\t}(\C)$ and $\marg_{S,\t}(\C)$ provide an upper and a lower bound, respectively, to the sum of the
utility functions over any set $\C$ of agents. This is particularly useful whenever we have to reason in terms of fairness for groups of
agents, rather than just in terms of the utility of singletons.

For instance, in the motivating scenario of Section~\ref{sec:motivation}, a crucial question concerns how the structure funding after the ANVUR
evaluation (e.g., the research funds for a University) should be shared among its sub-structures (e.g., the Departments). The result below
shows that, with our mechanism, any sub-structure will never get less than its marginal contribution, neither more than the maximum
contribution it can achieve if its members were alone in the structure. In particular, any closed group of researchers (e.g., any department
without collaborations with other departments, or any research group without further coauthors in the same structure) will share precisely the
total value attributed by ANVUR to its research products. Such a score is desirable for agents, and it is perceived as a fair distribution (see
the work by~\citeA{Moulin1992}, for more on the fairness of the Shapley value).

\begin{theorem}\label{thm:bounds}
Let $\pi$ be an optimal allocation for $S=\tuple{\A,G,\omega}$ w.r.t.~$\t$. Then, for each set $\C\subseteq \A$ of agents,
$\best_{S,\t}(\C)\geq \sum_{i\in \C}u_{i}(\pi,\t)\geq \marg_{S,\t}(\C)$.
\end{theorem}
\begin{proof}
By Theorem~\ref{thm:shapley3}, we know that $u_{i}(\pi,\t)=\sv_i(\game_{S,\t}^{\tbest})=\sv_i(\game_{S,\t}^{\tmarg})$, for each agent $i\in \A$
and optimal allocation $\pi$. Then, we can simply recall that the function $\marg_{S,\t}$ (resp., $\best_{S,\t}$) associated with the game
$\game_{S,\t}^{\tmarg}$ (resp., $\game_{S,\t}^{\tbest}$) is supermodular (resp., submodular) by Theorem~\ref{thm:supermodular} (resp.,
Proposition~\ref{prop:submodular}). Hence, the result follows as $\sum_{i\in \C}u_{i}(\pi,\t)=\sum_{i\in
\C}\sv_i(\game_{S,\t}^{\tbest})=\sum_{i\in \C}\sv_i(\game_{S,\t}^{\tmarg})$ and by property (II).~\hfill~
\end{proof}

Our second result pertains the budget-balance property of the mechanisms. Again, the correspondence with the Shapley value is crucial to establish the
result.

\begin{theorem}\label{thm:budget}
Let $\pi$ be an optimal allocation for $S=\tuple{\A,G,\omega}$ w.r.t.~$\t$. Then, it holds that $\sum_{i\in \A}p^\xi_i(\pi,\t)=0$.
\end{theorem}
\begin{proof}
By Theorem~\ref{thm:shapley3}, we know that $u_{i}(\pi,\t)=\sv_i(\game_{S,\t}^{\tbest})$, for each agent $i\in \A$ and optimal allocation
$\pi$, where $\sv_i(\game_{S,\t}^{\tbest})$ is the Shapley value of $\game_{S,\t}^{\tbest}$. By property (I) of the Shapley value, we know that
$\sum_{i\in \A}\sv_i(\game_{S,\t}^{\tbest})=\best_{S,\t}(\A)$. Thus, $\sum_{i\in \A}u_{i}(\pi,\t)=\sum_{i\in
\A}\sv_i(\game_{S,\t}^{\tbest})=\opt(\tuple{\A,G,\omega},\t)$.
It follows that $\opt(\tuple{\A,G,\omega},\t)=\sum_{i\in \A}v_i(\pi)-\sum_{i\in \A}p^\xi_i(\pi,\t)$, by definition of the utility. Hence,
$\sum_{i\in \A} p^\xi_i(\pi,\t)=\opt(\tuple{\A,G,\omega},\t)-\val(\pi,\t)=0$, as $\pi$ is indeed an optimal allocation w.r.t.~$\t$ (and, hence,
w.r.t. verified types).~\hfill~
\end{proof}

\begin{corollary}[\textbf{budget-balance}]
Let $A$ be any optimal allocation algorithm. Then, the mechanism with verification $(A,\p^\xi)$ is budget-balanced.
\end{corollary}

Finally, we complete the picture of our analysis by proving the strong fairness property of the proposed payment rule $\p^\xi$: In words, the
best outcome for every agent is always determined by a (global) optimal allocation. Moreover, from Corollary~\ref{cor:independence}, any agent
is indifferent about the specific optimal allocation being considered. That is, {\em any} chosen optimal allocation leads to the best results
for all agents.

\begin{lemma}\label{lemma:strong}
Let $\pi$ and $\pi'$ be two allocations for $S=\tuple{\A,G,\omega}$ such that $\pi$ is optimal, and hence $\val(\pi,\t)\geq \val(\pi',\t)$. Then, $u_{i}(\pi,\t)\geq
u_{i}(\pi',\t)$ holds, for each $i\in \A$.
Moreover, if $\pi'$ is not optimal, there exists some agent $i\in \A$ such that $u_{i}(\pi,\t) > u_{i}(\pi',\t)$.
\end{lemma}
\begin{proof}
For any allocation $\bar\pi$, consider the coalitional game $\game^{\bar\pi}=\tuple{\A,v^{\bar\pi}}$ such that
$v^{\bar\pi}(\C)=\opt(\tuple{\C,\img({\bar\pi}),\omega},\t)$, for each $\C\subseteq \A$. By looking at the expression of the Shapley value for
$\game^{\bar\pi}$, it is easy to check that $u_{i}({\bar\pi},\t)=\sv_i(\game^{\bar\pi})$ (just use the same reasoning leading to
Equation~\ref{eq:img} in the proof of Theorem~\ref{thm:shapley1}).
Assume now that $\pi'$ is an allocation with $\val(\pi,\t)\geq \val(\pi',\t)$, and consider the value
$v^{\pi'}(\C)=\opt(\tuple{\C,\img(\pi'),\omega},\t)$, for each $\C\subseteq \A$. By Corollary~\ref{cor:due}, we have that $v^\pi(\C)\geq
v^{\pi'}(\C)$, for each $\C\subseteq \A$. Then, we derive that $u_{i}(\pi,\t)=\sv_i(\game^\pi)\geq \sv_i(\game^{\pi'})=u_{i}(\pi',\t)$ for
every $i\in \A$, because of property (III) of the Shapley value.

Now assume that $\pi'$ is not optimal, and thus $\val(\pi,\t) > \val(\pi',\t)$. Therefore, for the grand-coalition $\A$, we have $v^\pi(\A) >
v^{\pi'}(\A)$. Because of property (I) of the Shapley value, only (and all) the total value $v^{\pi'}(\A)$ is distributed to agents. It follows
that there exists some agent $i\in \A$ such that $u_{i}(\pi,\t)=\sv_i(\game^\pi)> \sv_i(\game^{\pi'})=u_{i}(\pi',\t)$.~\hfill~
\end{proof}

Because, truthful declarations lead to optimal allocations, the desired fairness property is immediately entailed by the previous lemma.

\begin{theorem}[\textbf{Fairness}]\label{theo:fairness}
Let $A$ be any optimal allocation algorithm. Then, for any agent $i\in \A$ and any allocation $\pi$, $u_{i}(A(\t),\t)\geq u_{i}(\pi,\t)$.
\end{theorem}

\begin{corollary}[\textbf{Pareto-efficiency and envy-freeness}]
Let $A$ be any optimal allocation algorithm. Then, the mechanism with verification $(A,\p^\xi)$ is Pareto efficient and envy-free.
\end{corollary}

Note that the above fairness condition guarantees much more than classical Pareto efficiency and envy-freeness, because it entails that the
mechanism leads to a unique evaluation, independently of the chosen optimal allocation. In particular, the Pareto set is a singleton.

\section{Complexity Issues}\label{sec:complexity}

In this section, we shall reconsider our mechanism with verification from a computational perspective. Note first that computing an optimal
allocation on the basis of the reported types is an easy task, which can be carried out via adaptations of classical matching algorithms.
Indeed, in the light of Fact~\ref{lemma:norm}, computing an optimal allocation for $\tuple{\A,G,\omega,\w}$ reduces to computing an optimal
allocation for $\tuple{\A^{\bf 1},G,\bf{1}}$ w.r.t.~$\w^{\bf 1}$, which is a scenario where each agent can be allocated one good at most. This
is equivalent to find a matching of maximum weight over a complete bipartite graph over the set of disjoint nodes $\A^{\bf 1}$  and $G$, and
where edge weights are encoded via the function $\w^{\bf 1}$. This task is well-known to be feasible in polynomial time~\cite<e.g.,
>{Schrijver2003}.

\subsection{Hardness Result}

Despite optimal allocations can be computed in polynomial time, our mechanism is not computationally-efficient, since payments are unlikely to
be computable in polynomial time. Indeed, we next show that this computation problem is complete for the complexity class $\rm \#P$~\cite<see
>{Papadimitriou1993}.

For the sake of completeness, we recall here that a \emph{counting Turing machine} is a standard nondeterministic Turing machine with an
auxiliary output device that prints in binary notation the number of accepting computations induced by the input. It has (worst-case) time
complexity $f(n)$ if the longest accepting computation induced by the set of all inputs of size $n$ takes $f(n)$ steps. Then, {\rm \#P} is the
class of all functions that can be computed by counting Turing machines of polynomial time complexity. A prototypical {\rm \#P}-\emph{complete}
problem is to count the number of truth variable assignments that satisfy a Boolean formula. Of course, {\rm NP}$\subseteq${\rm \#P}, and a
polynomial-time algorithm for solving a {\rm \#P}-complete problem would imply {\rm P} = {\rm NP}.

\begin{theorem}
Computing the Shapley value of coalitional games associated with allocation problems (as in Definition~\ref{def:games}) is {\rm
\#P}-{complete}.
\end{theorem}
\begin{proof}
The problem belongs to {\rm \#P}, because computing the Shapley value is known to be feasible in {\rm \#P} for any class of coalitional games
with polynomial-time value/cost functions~\cite<c.f. >{Deng1994}. To show that it is {\rm \#P}-{hard}, we exhibit a reduction from the
following problem:
Let ${\tt G}=(A\cup B,E)$ be a bipartite graph with $|A|=|B|=n$, $E\subseteq A\times B$, and $|E|=m\geq n$. Recall that a matching is a set
$E'\subseteq E$ of edges such that for each pair of distinct edges $(a,b)$ and $(a',b')$ in $E'$, $a\neq a'$ and $b\neq b'$ hold. The matching
$E'$ is \emph{perfect} if $|E'|=n$. The problem of counting the number of perfect matchings in such bipartite graphs is {\rm
\#P}-complete~\cite{Valiant1979}.

Given a graph ${\tt G}=(A\cup B,E)$ as above and a constant $k\geq 1$ (which we shall fix below), we build in polynomial-time a tuple $S({\tt
G})=\tuple{\A,G,\omega}$ and a type vector $\t$ such that:

\begin{myitemize}
\item[\hspace{-2mm}(1)] $\A=\{\alpha\}\cup\bigcup_{(a,b)\in E}\{(a,b)^1,...,(a,b)^k\}$, i.e., agents are one-to-one associated with $k$
    distinct clones of each edge $(a,b)\in E$, plus a distinguished node $\alpha$. Note that $|\A|>n$, because in the considered bipartite graphs $m\geq n$ holds;
\item[\hspace{-2mm}(2)] $G=\{g_\alpha\}\cup A\cup B$, i.e., goods correspond to nodes, plus a distinguished good $g_{\alpha}$;
\item[\hspace{-2mm}(3)] $w$ is the function such that $\omega(\alpha)=1$, and $\omega((a,b)^i)=2$, for each $(a,b)\in A$ and
    $i\in\{1,...,k\}$;
\item[\hspace{-2mm}(4)] Types are as follows. For each $(a,b)^i\in \A$, $t_{(a,b)^i}(a)=2$, $t_{(a,b)^i}(b)=2$,
    $t_{(a,b)^i}(g_{\alpha})=1$, $t_{(a,b)^i}(x)=0$, $\forall x\in (A\cup B)\setminus\{a,b\}$.
    Moreover, $t_{\alpha}(g_{\alpha})=1$, $t_{\alpha}(x')=0$,  $\forall x'\in (A \cup B)$.
\end{myitemize}

Let us now fix some notations. For any set $E'\subseteq E$ of edges, let $\mathtt{match}(E')$ denote the size of the largest set $E''\subseteq
E'$ of edges that is a matching. For any set $\C\subseteq \A\setminus\{\alpha\}$ of agents, let $\mathtt{A}(\C)=\{ a \mid (a,b)^i \in \C\}$ and
$\mathtt{B}(\C)=\{ b \mid (a,b)^i \in \C\}$.
Finally, we say that $\C\subseteq \A\setminus\{\alpha\}$ is \emph{tight} if it does not contain two agents of the form $(a,b)^i$ and $(a,b)^j$,
with $i\neq j$, i.e., associated with the same edge of $\tt G$.

Observe that, for each set $\C\subseteq \A\setminus\{\alpha\}$ of agents,
\begin{eqnarray}\label{eqn:rew}
\opt(\tuple{\C\cup \{\alpha\},G,\omega},\t)-\opt(\tuple{\C,G,\omega},\t)=\left\{
\begin{array}{ll}
1 & \mbox{ if $\C$ is tight, and $|\C|=\mathtt{A}(\C)=\mathtt{B}(\C)$}         \\
0 & \mbox{ otherwise}
\end{array}
\right.
\end{eqnarray}
Indeed, if $\pi_\C'$ is an optimal allocation for $\tuple{\C\cup\{\alpha\},G,{\omega}}$ w.r.t.~$\t$, then we always have that
$\val(\pi_\C',\t)=2\times |\mathtt{A}(\C)|+2\times |\mathtt{B}(\C)| + 1$.
Instead, if $\pi_\C$ is an optimal allocation for $\tuple{\C,G,{\omega}}$ w.r.t.~$\t$, then we have
$$
\val(\pi_\C)=\left\{
\begin{array}{ll}
2\times |\mathtt{A}(\C)|+2\times |\mathtt{B}(\C)| & \mbox{ if $\C$ is tight, and $|\C|=\mathtt{A}(\C)=\mathtt{B}(\C)$}         \\
2\times |\mathtt{A}(\C)|+2\times |\mathtt{B}(\C)|+1 & \mbox{ otherwise}
\end{array}
\right.
$$

By exploiting Equation~\ref{eqn:rew}, we can now express the Shapley value of the game $\game^{\tbest}_{S(\tt G),\t}$ for agent $\alpha$ in a
convenient way. Let $X_h$ denote the number of sets $\C\subseteq \A\setminus\{\alpha\}$ of agents which are tight and such that $|\C|=|{\tt
A}(\C)|=|\mathtt{B}(\C)|=h$, and let $X_0=1$. Then,
\begin{eqnarray}
\sv_\alpha(\game^{\tbest}_{S(\tt G),\t})=\sum_{h=0}^{|\A|-1}\frac{(|\A|-h-1)!(h)!}{|\A|!} X_h.
\end{eqnarray}
In particular, let us now focus on the coefficient $X_h$. Denote by $Y_h$ the number of matchings in $\tt G$ whose cardinality is $h$. By
construction of $S({\tt G})$ it is immediate to check that for each matching of cardinality $h$ in $\tt G$, there are precisely $k^h$ sets of
agents $\C\subseteq \A\setminus\{\alpha\}$ that are tight and such that $|\C|=|{\tt A}(\C)|=|\mathtt{B}(\C)|=h$. Thus, we can rewrite the above
expression:

\vspace{-6mm}
\begin{eqnarray}
\sv_\alpha(\game^{\tbest}_{S(\tt G),\t})=\sum_{h=0}^{|\A|-1}\left(Z_h \times Y_h\right)\times k^h, \mbox{ with $Z_h=\frac{(|\A|-h-1)!(h)!}{|\A|!}$}.
\end{eqnarray}

For an expression as the one above, given the value of $\sv_\alpha(\game^{\tbest}_{S(\tt G),\t})$, it is known that under certain circumstances
we can reconstruct in polynomial time the value of each single term of the form $Z_h \times Y_h$ (see Fact 6 in the work
by~\citeA{Valiant1979a}): We need the existence of an integer constant ${\bf A}>2$ such that, for each $h\in\{0,...,|\A|-1\}$, $Z_h \times
Y_h\leq {\bf A}$, and $k\geq {\bf A}^2$.
In our case, it can be noticed that, for each $h\in\{0,...,|\A|-1\}$, $Z_h \times Y_h\leq 1$ holds, as $Y_h\leq |\A|!/({(h)!(|\A|-h)!})$. Thus,
for $k=9$, we have that, given the value of $\sv_\alpha(\game^{\tbest}_{S(\tt G),\t})$, we can compute in polynomial time all such terms.
In particular, we can compute in polynomial time the term associated to $h=|A|=|B|=n$, where recall that  $|\A|>n$. This term has
the form $Z_n \times Y_n$, with $Y_n$ being the number of perfect matchings in $\tt G$. Thus, by putting it all together and since $Z_n$ can be
computed in polynomial time (as the size of the numbers $n$ and $|\A|$ are logarithmic w.r.t.~the size of $\tt G$), the number of perfect
matchings in bipartite graphs can be counted in polynomial time too, which concludes the proof.~\hfill~
\end{proof}

By Lemma~\ref{lemma:utility} and Theorem~\ref{thm:shapley3}, the following is immediate.

\begin{corollary}
Computing the payments as given by the rule $\p^\xi$ is $\rm \#P$-complete.
\end{corollary}

\subsection{A Fully Polynomial-Time Randomized Approximation Scheme}

An approach to circumvent the intractability of the Shapley value is based on approximation: For a game $\game=\tuple{N,\varphi}$, a vector
$\hat \sv$ is an \emph{$\varepsilon$-approximation of the Shapley value} if $|\hat \sv_i-\sv_i(\game)|\leq \varepsilon \times \sv_i(\game)$
holds, for each $i\in N$.

Recently, a sampling method conceived by \citeA{Bachrach2010} for the special class of \emph{simple} coalitional games has been extended to
deal with arbitrary games that are {supermodular} and \emph{monotone}\footnote{Monotonicity of $\game$ means that $\varphi(R)\geq \varphi(T)$,
$\forall T\subseteq R\subseteq N$.}~\cite{Liben-Nowell2011},  under the assumption that the value $\varphi(R)$ can be computed by an oracle
having unitary cost, for each $R\subseteq N$. The result is that, for any $\varepsilon>0$ and $\delta>0$, it is possible to compute in time
${\tt poly}(N,1/\varepsilon,\log(1/\delta))$ a vector $\hat \sv$ that is an $\varepsilon$-approximation of the Shapley value with probability
of failure at most $\delta$. A method with this properties is called a fully polynomial-time randomized approximation scheme.

\begin{figure}[t]
\centering  \fbox{
\parbox{0.99\textwidth}
{
\begin{tabular}{ll}
\textbf{Input}:     &\ An allocation $\pi$ for $\tuple{\A,G,\omega}$, a vector $\w\in \D$, and an integer $m>0$;  \\
\textbf{Assumption}: &\ A verifier $\v$ is available. Let $\v(\pi)=(v_1,...,v_n)$;
\end{tabular}\\
\vspace{-0mm}\hrule\vspace{1mm}
\begin{tabular}{ll}
   1.\ \ \ & Generate a set $\hmymath{C}$ of $m$ subsets of $\A$, and add to them
   the grand-coalition $\A$; \\
   2.\ \ \ & For each set $\C\in \hmymath{C}$,\\
   3.\ \ \ & $|$\ \ \ \ Compute an optimal allocation $\pi_\C$ for $\tuple{\C,\img(\pi),\onC{\omega}}$ w.r.t.~$\w$; \\
   4.\ \ \ & \hspace{-0.2mm}$\lfloor$\ \ \ \ Compute an optimal allocation $\pi_{\C\setminus\{i\}}$ for $\tuple{\C\setminus\{i\},\img(\pi),\onCi{\omega}}$ w.r.t.~$\w$; \\
   5.\ \ \ & For each agent $i\in \A$,\\
   6.\ \ \ & $\lfloor$\ \ \ \ Compute $\xi_i(\pi,\w)$ as in Figure~\ref{fig:M1} (steps 4---8), with $\mymath{C}:=\hmymath{C}$; \\
   7.\ \ \ & Repeat $\Theta(\log(1/\delta))$ times steps 1, 2, and 5, and\\
   8.\ \ \ & Let $\hat \xi(\pi,\w)$ be the component-wise median vector of these vectors $\xi(\pi,\w)$;\\
   9.\ \ \ & Define $\hat p_i^\xi(\pi,\w):=\hat \xi_i(\pi,\w)-v_i(\pi)$;    \\
\end{tabular}
} }\vspace{-1mm}
  \caption{Payment rule $\hat \p^\xi$.} \label{fig:M2}
\end{figure}

Next, we propose a payment rule $\hat \p^\xi$ that is founded on the sampling strategy described in the work by~\citeA{Liben-Nowell2011}. The
payment rule, reported in Figure~\ref{fig:M2}, samples $m$ subsets of $\A$ storing them in $\hmymath{C}$, and then computes the value
$\xi(\pi,\w)$ as in Figure~\ref{fig:M1}, but with $\hmymath{C}$ playing the role of the power-set $\mymath{C}$. Eventually, the process is repeated
$\Theta(\log(1/\delta))$ times, and the component-wise median vector of all such payments is computed. Finally, at step~9, the usual
compensation and bonus approach is implemented.

Interestingly, though the new rule $\hat \p^\xi$ is based on randomization, the following properties still hold (always, not just as expected
outcomes).

\begin{theorem}\label{thm:truth-approx}
Let $A$ be any optimal allocation algorithm. Then, the mechanism with verification $(A,\hat \p^\xi)$ is truthful and individually-rational.
\end{theorem}
\begin{proof}
The result follows by inspecting the proofs for rule $\p^\xi$ in Section~\ref{sec:mechanism}. Indeed, it can be immediately checked that those
proofs do not depend on the specific subset of coalitions $\mymath{C}$, and thus they smoothly apply if any set of coalitions $\hmymath{C}$ is
used as in Figure~\ref{fig:M2}, instead of all possible subsets of $\A$.
Note in particular that, despite the payment rule is based on randomization, the resulting mechanism is always truthful: 
just look at the proof of Theorem~\ref{thm:truth}, and notice that properties (A) and (B) are precisely those guaranteeing truthfulness, and
that these properties hold for each given coalition $\C$. Therefore, they still hold for any subset of coalitions randomly chosen by the
mechanism.~\hfill~
\end{proof}

For continuing with a deeper analysis of the payment rule $\hat \p^\xi$, we need to point out a relationship between utility values and
approximations of the Shapley value.

\begin{lemma}\label{lemma:shapley3-approx}
Let $\A=\{1,...,{|\A|}\}$, and let $m=\Theta(|\A|^2/\varepsilon^2)$. Then, for each optimal allocation $\pi$ for $S$ w.r.t.~$\t$, the vector
$(u_{1,\hat \p^\xi}(\pi,\t),...,u_{{|\A|},\hat \p^\xi}(\pi,\t))$ is in expectation the Shapley value of $\game_{S,\t}^{\tmarg}$ (and
$\game_{S,\t}^{\tbest}$), of which it is
an $\varepsilon$-approximation, with probability $1-\delta$.
\end{lemma}
\begin{proof}
By exploiting the same line of reasoning as in the proofs in Section~\ref{sec:mechanism} for $\p^\xi$, we can see that $\xi_i(\pi,\t)$ (at
step~6 of the algorithm in Figure~\ref{fig:M2}) can be rewritten as follows:

\vspace{-3mm}
$$\xi_i(\pi,\t)=\sum_{\C \in \hmymath{C}}\frac{(|\A|-|\C|)!(|\C|-1)!}{|\A|!}\left(\marg_{S,\t}(\C)-\marg_{S,\t}(\C\setminus\{i\})\right).$$

Now, recall that the game $\game_{S,\t}^{\tmarg}=\tuple{\A,\marg_{S,\t}}$ is supermodular by Theorem~\ref{thm:supermodular}. Moreover,
$\game_{S,\t}^{\tmarg}$ is clearly monotone. Thus, by Theorem~4 in the work by~\citeA{Liben-Nowell2011}, we derive the result in the statement,
but for the fixed value of $\delta=1/4$. Steps 7 and 8 just serve to amplify the probability~\cite<c.f.>{Liben-Nowell2011}, and to get a fully
polynomial-time randomized approximation scheme.~\hfill~
\end{proof}

As the expected utility profile coincides with the Shapley value, it is easy to see that $\hat \p^\xi$ enjoys in expectation all the properties
of $\p^\xi$ (e.g., Pareto-efficiency and envy-freeness). We thus focus in this section on those properties that can be shown to hold always,
i.e., not just in expectation.
With this respect, note that in the approach by~\citeA{Liben-Nowell2011}, a final normalization step is carried out to preserve the budget
balance. Unfortunately, this way truthfulness might be lost, hence we did not include such a normalization procedure in the above payment rule.
As a consequence, the mechanism $\hat \p^\xi$ does not guarantee budget-balancedness and Pareto-efficiency. However, we can still have
approximate counterparts for Theorem~\ref{thm:bounds} and Theorem~\ref{thm:budget}.

\begin{theorem}\label{thm:bounds-bis}
Let $\pi$ be an optimal allocation for $S$ w.r.t.~$\t$. Let $m=\Theta(|\A|^2/\varepsilon^2)$. Then, with probability $1-\delta$,
\begin{myitemize}
\item[$\bullet$] $(1+\varepsilon)\times\best_{S,\t}(\C)\geq \sum_{i\in \C}u_{i,\hat \p^\xi}(\pi,\t)\geq (1-\varepsilon)\times\marg_{S,\t}(\C)$,
    for
    each
    $\C\subseteq \A$;

\item[$\bullet$] $\varepsilon \times \val(\pi,\t) \geq \sum_{i\in \A}\hat p^\xi_i(\pi,\t)\geq
 -\varepsilon\times \val(\pi,\t)$.
\end{myitemize}
\end{theorem}
\begin{proof}
Here, just observe that, in the light of Lemma~\ref{lemma:shapley3-approx} and Theorem~\ref{thm:shapley3}, for each set $\C\subseteq \A$, we
have $(1-\varepsilon)\times \sum_{i\in \C}u_{i,\hat \p^\xi}(\pi,\t)\leq \sum_{i\in \C}u_{i,\p^\xi}(\pi,\t)\leq (1+\varepsilon)\times \sum_{i\in
\C}u_{i,\hat \p^\xi}(\pi,\t)$. The result then follows by substituting such bounds in Theorem~\ref{thm:bounds} and Theorem~\ref{thm:budget},
respectively, with simple algebraic manipulations.
\end{proof}

Finally, we propose a further randomized mechanism that is able to guarantee both economic efficiency and budget-balancedness. The price to be
paid is however that truthfulness holds in {expectation} only.
The mechanism is based on a payment rule $\bar \p^\xi$.

\begin{theorem}
Let $A$ be any optimal allocation algorithm. Then, the (randomized) mechanism with verification $(A,\bar \p^\xi)$ is truthful in expectation.
Moreover, (at the truthful equilibrium) it is efficient, individually-rational and budget-balanced.
\end{theorem}
\begin{proof}
The payment rule $\bar \p^\xi$ follows the steps in Figure~\ref{fig:M2}, with minor modifications at step 8 and step 9: First, at step 8,
whenever we compute the median value $\hat \xi_i(\pi,\w)$ for agent $i$, we also compute the corresponding value $\hat \xi_i(\pi,\v(\pi))$
(evaluated on the revealed types rather than on the reported ones). Then, we define a normalization factor
$R=\opt(\tuple{\A,\img(\pi),\omega},\v(\pi))/(\sum_{i\in \A} \hat \xi_i(\pi,\v(\pi)))$, so that, at step 9, $\bar p_i^\xi(\pi,\w)$ is
eventually returned as $v_i(\pi)-\hat \xi_i(\pi,\w)\times R$.

Concerning truthfulness, we can just note that the expected value of $R$ is 1. Indeed, by Lemma~\ref{lemma:shapley3-approx}, the expected value
of $\hat \xi_i(\pi,\v(\pi))$ is $\xi_i(\pi,\v(\pi))$; hence, the sum of all these values coincides with
$\opt(\tuple{\A,\img(\pi),\omega},\v(\pi))$ by the efficiency of the Shapley value (as in the proof of Theorem~\ref{thm:budget}). Thus, the
expected utility of an agent $i$ under the payment rule $\bar \p^\xi$ coincides with the (actual, i.e., not in expectation) utility of $i$
under the rule $\hat \p^\xi$. Hence, truthfulness in expectation follows by Theorem~\ref{thm:truth-approx}.
Now, we can just check that, at the truthful equilibrium, the maximum social welfare is achieved (equilibrium efficiency) and  $\sum_{i\in
\A}\bar p^\xi_i(\pi,\t)=\opt(\tuple{\A,\img(\pi),\omega},\v(\pi))-\sum_{i\in\A}\hat \xi_i(\pi,\w)\times R=0$. That is, the mechanism is
budget-balanced, too.
Finally, the mechanism $\bar \p^\xi$ is seen to be individually-rational, by exploiting the same line of reasoning as the one used for the
mechanism based on $\p^\xi$, since the corresponding proof in Section~\ref{sec:mechanism} is not affected by the sampling strategy.~\hfill~
\end{proof}

\section{Related approaches to Mechanisms with Verification}\label{sec:related}

We next review the main approaches in the literature for mechanisms with verification.

In the works by \citeA{Auletta2009,Penna2009,Krysta2010} and \citeA{Ferrante2009}, the individual welfare of an agent $i$, given the outcome
$\pi$ and the vector $\d$ of reported types, is assumed to be of the following form:

\vspace{-2mm}
$$u_{i,{\p}}(\pi,\d)=t_i(\pi)-\left\{
\begin{array}{ll}
0 & \mbox{if $i$ is caught lying}\\
p_i(\pi,\d) & \mbox{otherwise}\\
\end{array}
\right.$$

\noindent where $p_i(\cdot,\cdot)$ is a payment that does not depend on the vector $\t$ of the true types.

In these papers, the only information that is assumed to be available at payment time is whether the reported type $d_i$ of agent $i$ differs
or not from its actual true type $t_i$, so that the knowledge of $t_i$ is basically immaterial. On the other hand,
the specific payment scheme adopted punishes those agents that are caught lying. Therefore, while the verification process provides a smaller
amount of information than the verification process in our approach, the rules used to discourage strategic behaviors are stronger than ours
and based on punishing agents. Moreover, the above works assume that agents' misreporting is restricted only to certain kinds of lies (e.g.,
values lower than the corresponding true ones), so that a form of ``one-sided'' verification suffices.

Recently, the above model of (partial) verification has been extended by \citeA{Caragiannis2012} to a setting where an agent cheating on
her/his type will be identified with some probability that may depend on her/his true type, the reported type, or both. The payment scheme is
exactly the same as the one discussed above and, hence, verification does not exploit the knowledge of the actual true type and a punishment
approach is still used. The main novelty, in addition to the probabilistic verification, is that there is no constraint on the type that an
agent can report while cheating.

Finally, a different kind of verification model goes back to the seminal paper by \citeA{Nisan2001}, and is actually closer to our
``no-punishment'' perspective, because an agent $i$ can in principle be paid by the mechanism even if $i$ has been caught lying. Given $n$
agents, \citeA{Nisan2001} consider a vector $\mathbf{e}=(e_1,...,e_n)$ of ``observed'' agent types, which are completely known after the
verification process. Moreover, the individual utility of any agent $i$ has the form:

\vspace{-4mm}
$$u_{i,{\p}}(\pi,\d)=e_i(\pi)-p_i(\pi,\d),$$
so that the vector $\mathbf{e}$ in such a framework plays the same role as the vector of verified types in our approach.
Note that in some settings it does not make sense to assume that the utility of an agent depends on verified/observed types
\cite<cf.>{Penna2009}. In our motivating scenario, this is not the case, as the funds received by researchers are precisely determined by the
verifier (i.e., ANVUR evaluation), here playing ideally the role of providing an ``objective'' utility to the agents or, putting it in more
pragmatic terms, an utility determined by external constraints---in fact by law.

A first difference between the work by \citeA{Nisan2001} and our approach is that, in the above model, agents' misreporting is again restricted
only to certain kind of lies.
Another more subtle difference is that our verification process can be defined as a {\em good-centric} one, because at payment time everything
is known about each verified good. More precisely, if good $g$ has been verified, then we know everything about its value, that is, its actual
evaluation $t_i(g)$ according to each agent $i$, even if $g$ is not allocated to $i$ (recall that our mechanism is in fact indifferent with
respect to alternative allocations).
Instead, the setting by \citeA{Nisan2001} can be viewed as an {\em agent-centric} one, because the true type of each agent $i$ for the goods
allocated to $i$ are disclosed by the verification process.

It is easy to see that a good-centric verification provides more information, in general. On the other hand, this additional information turns
out to be the crucial feature to overcome classical impossibility results, and to meet all desirable properties at once (without using any
punishing power).
Moreover, it is worthwhile noting that good-centric and agent-centric verifications have the same power on all classes of problems where
$t_i(g)=t_j(g)$, for every good $g\in G$ and each pair of agents $i,j\in \A$ with positive valuations about $g$, that is, whenever the
application at hand is such that the value of a good is an objective property and hence, if the good is verified, it is disclosed for all
agents at once. For instance, as discussed in Section~\ref{sec:framework}, our motivating scenario about the 2012 Italian research assessment
programme is precisely of this form.

\section{Conclusion}\label{sec:conclusion}

In this paper, we have proposed and analyzed mechanisms for fair allocation problems. In classical settings, it is well-known that there is no
truthful mechanism that can be simultaneously efficient, budget-balanced, and fair. Here, motivated by a real-world problem, we have considered
mechanisms with {verification}, where payments to agents can be performed {after} the final outcome is known and verified.  In particular, we
have proposed a model of verification that is able to disclose the true values of allocated goods, in contrast to previous approaches in the
literature where partial and probabilistic verification have been considered. However, the use of this verification power is in fact quite
limited because payment rules have been designed without punishing in any way those agents that are possibly caught lying. The resulting
mechanisms have been analyzed by taking into account both algorithmic and complexity issues.

By looking at the proposed framework from an abstract perspective, one may notice that it is based on two fundamental ingredients: a base
combinatorial problem that determines feasible and optimal allocations, and a game-theoretic notion that describes what is considered fair,
with respect to agents' contributions and expectations. In the application domain addressed in the paper, it was natural to consider the
weighted matching as the basic combinatorial problem and the Shapley value as the game-theoretic solution concept. In fact, an interesting
avenue of further research is to study different instances of such an abstract framework for mechanisms with verification, where other
combinatorial problems (colorings, coverings, etc.) and different solution concepts (Nucleolus, Banzhaf index, etc.) may be more appropriate
and best describe the problem at hands.

\bibliographystyle{theapa}
\bibliography{MD}

\begin{thebibliography}{}

\bibitem[\protect\BCAY{Abdulkadiro{\u g}lu, S{\"o}nmez,\ \BBA\
  {\"U}nver}{Abdulkadiro{\u g}lu et~al.}{2004}]{Abdulkadiroglu2004}
Abdulkadiro{\u g}lu, A., S{\"o}nmez, T., \BBA\ {\"U}nver, M.~U. \BBOP2004\BBCP.
\newblock \BBOQ Room assignment-rent division: A market approach\BBCQ\
\newblock {\Bem Social Choice and Welfare}, {\Bem 22}, 515--538.

\bibitem[\protect\BCAY{Alcalde\ \BBA\ Barber{\`a}}{Alcalde\ \BBA\
  Barber{\`a}}{1994}]{Alcalde1994}
Alcalde, J.\BBACOMMA\  \BBA\ Barber{\`a}, S. \BBOP1994\BBCP.
\newblock \BBOQ Top dominance and the possibility of strategyproof stable
  allocations to matching problems\BBCQ\
\newblock {\Bem Economic Theory}, {\Bem 4}, 417--435.

\bibitem[\protect\BCAY{Alkan, Demange,\ \BBA\ Gale}{Alkan
  et~al.}{1991}]{Alkan1991}
Alkan, A., Demange, G., \BBA\ Gale, D. \BBOP1991\BBCP.
\newblock \BBOQ Fair allocation of indivisible goods and criteria of
  justice\BBCQ\
\newblock {\Bem Econometrica}, {\Bem 59\/}(4), 1023--39.

\bibitem[\protect\BCAY{Andersson}{Andersson}{2009}]{Andersson2009}
Andersson, T. \BBOP2009\BBCP.
\newblock \BBOQ A general strategy-proof fair allocation mechanism
  revisited\BBCQ\
\newblock {\Bem Economics Bulletin}, {\Bem 29\/}(3), 1717--1722.

\bibitem[\protect\BCAY{Andersson\ \BBA\ Svensson}{Andersson\ \BBA\
  Svensson}{2008}]{Andersson2008}
Andersson, T.\BBACOMMA\  \BBA\ Svensson, L.-G. \BBOP2008\BBCP.
\newblock \BBOQ Non-manipulable assignment of individuals to positions
  revisited\BBCQ\
\newblock {\Bem Mathematical Social Sciences}, {\Bem 56\/}(3), 350--354.

\bibitem[\protect\BCAY{Andersson, Svensson,\ \BBA\ Ehlers}{Andersson
  et~al.}{2010}]{Andersson2010}
Andersson, T., Svensson, L.-G., \BBA\ Ehlers, L. \BBOP2010\BBCP.
\newblock \BBOQ Budget-balance, fairness and minimal manipulability\BBCQ\
\newblock Working papers\ 2010:16, Lund University, Department of Economics.

\bibitem[\protect\BCAY{Aragones}{Aragones}{1995}]{Aragones1995}
Aragones, E. \BBOP1995\BBCP.
\newblock \BBOQ A derivation of the money rawlsian solution\BBCQ\
\newblock {\Bem Social Choice and Welfare}, {\Bem 12}, 267--276.

\bibitem[\protect\BCAY{Archer\ \BBA\ Tardos}{Archer\ \BBA\
  Tardos}{2007}]{Archer2007}
Archer, A.\BBACOMMA\  \BBA\ Tardos, E. \BBOP2007\BBCP.
\newblock \BBOQ Frugal path mechanisms\BBCQ\
\newblock {\Bem ACM Transactions on Algorithms}, {\Bem 3}, 3:1--3:22.

\bibitem[\protect\BCAY{Auletta, De~Prisco, Penna,\ \BBA\ Persiano}{Auletta
  et~al.}{2009}]{Auletta2009}
Auletta, V., De~Prisco, R., Penna, P., \BBA\ Persiano, G. \BBOP2009\BBCP.
\newblock \BBOQ The power of verification for one-parameter agents\BBCQ\
\newblock {\Bem Journal of Computer and System Sciences}, {\Bem 75}, 190--211.

\bibitem[\protect\BCAY{Bachrach, Markakis, Resnick, Procaccia, Rosenschein,\
  \BBA\ Saberi}{Bachrach et~al.}{2010}]{Bachrach2010}
Bachrach, Y., Markakis, E., Resnick, E., Procaccia, A.~D., Rosenschein, J.~S.,
  \BBA\ Saberi, A. \BBOP2010\BBCP.
\newblock \BBOQ Approximating power indices: theoretical and empirical
  analysis\BBCQ\
\newblock {\Bem Autonomous Agents and Multi-Agent Systems}, {\Bem 20},
  105--122.

\bibitem[\protect\BCAY{Bevi{\'a}}{Bevi{\'a}}{1998}]{Bevia1998}
Bevi{\'a}, C. \BBOP1998\BBCP.
\newblock \BBOQ Fair allocation in a general model with indivisible goods\BBCQ\
\newblock {\Bem Review of Economic Design}, {\Bem 3}, 195--213.

\bibitem[\protect\BCAY{Brams\ \BBA\ Kilgour}{Brams\ \BBA\
  Kilgour}{2001}]{Brams2001}
Brams, S.~J.\BBACOMMA\  \BBA\ Kilgour, D.~M. \BBOP2001\BBCP.
\newblock \BBOQ Competitive fair division\BBCQ\
\newblock {\Bem Journal of Political Economy}, {\Bem 109\/}(2), 418--443.

\bibitem[\protect\BCAY{Brandt\ \BBA\ Endriss}{Brandt\ \BBA\
  Endriss}{2012}]{Brandt2012}
Brandt, F., C.~V.\BBACOMMA\  \BBA\ Endriss, U. \BBOP2012\BBCP.
\newblock {\Bem Multiagent Systems}, \BCH\ Computational Social Choices.
\newblock MIT Press.

\bibitem[\protect\BCAY{Caragiannis, Elkind, Szegedy,\ \BBA\ Yu}{Caragiannis
  et~al.}{2012}]{Caragiannis2012}
Caragiannis, I., Elkind, E., Szegedy, M., \BBA\ Yu, L. \BBOP2012\BBCP.
\newblock \BBOQ Mechanism design: from partial to probabilistic
  verification\BBCQ\
\newblock In {\Bem Proceedings of the 13th annual ACM Conference on Electronic
  Commerce}, EC'12. to Appear.

\bibitem[\protect\BCAY{Clarke}{Clarke}{1971}]{Clarke1971}
Clarke, E. \BBOP1971\BBCP.
\newblock \BBOQ Multipart pricing of public goods\BBCQ\
\newblock {\Bem Public Choice}, {\Bem 8}, 19–33.

\bibitem[\protect\BCAY{Deng\ \BBA\ Papadimitriou}{Deng\ \BBA\
  Papadimitriou}{1994}]{Deng1994}
Deng, X.\BBACOMMA\  \BBA\ Papadimitriou, C.~H. \BBOP1994\BBCP.
\newblock \BBOQ On the complexity of cooperative solution concepts\BBCQ\
\newblock {\Bem Mathematics of Operations Research}, {\Bem 19}, 257--266.

\bibitem[\protect\BCAY{Ferrante, Parlato, Sorrentino,\ \BBA\ Ventre}{Ferrante
  et~al.}{2009}]{Ferrante2009}
Ferrante, A., Parlato, G., Sorrentino, F., \BBA\ Ventre, C. \BBOP2009\BBCP.
\newblock \BBOQ Fast payment schemes for truthful mechanisms with
  verification\BBCQ\
\newblock {\Bem Theoretical Computer Science}, {\Bem 410}, 886--899.

\bibitem[\protect\BCAY{Green\ \BBA\ Laffont}{Green\ \BBA\
  Laffont}{1977}]{Green1977}
Green, J.\BBACOMMA\  \BBA\ Laffont, J. \BBOP1977\BBCP.
\newblock \BBOQ Characterization of satisfactory mechanisms for the revelation
  of preferences for public goods\BBCQ\
\newblock {\Bem Econometrica}, {\Bem 45\/}(2), 427--438.

\bibitem[\protect\BCAY{Green\ \BBA\ Laffont}{Green\ \BBA\
  Laffont}{1986}]{Green1986}
Green, J.~R.\BBACOMMA\  \BBA\ Laffont, J.-J. \BBOP1986\BBCP.
\newblock \BBOQ Partially verifiable information and mechanism design\BBCQ\
\newblock {\Bem The Review of Economic Studies}, {\Bem 53}, 447--256.

\bibitem[\protect\BCAY{Groves}{Groves}{1973}]{Groves1973}
Groves, T. \BBOP1973\BBCP.
\newblock \BBOQ Incentives in teams\BBCQ\
\newblock {\Bem Econometrica}, {\Bem 41}, 617–631.

\bibitem[\protect\BCAY{Haake, Raith,\ \BBA\ Su}{Haake et~al.}{2002}]{Haake2002}
Haake, C.-J., Raith, M.~G., \BBA\ Su, F.~E. \BBOP2002\BBCP.
\newblock \BBOQ Bidding for envy-freeness: A procedural approach to n-player
  fair-division problems\BBCQ\
\newblock {\Bem Social Choice and Welfare}, {\Bem 19\/}(4), 723--749.

\bibitem[\protect\BCAY{Hurwicz}{Hurwicz}{1975}]{Hurwicz1975}
Hurwicz, L. \BBOP1975\BBCP.
\newblock \BBOQ On the existence of allocation systems whose manipulative nash
  equilibria are pareto optimal\BBCQ.

\bibitem[\protect\BCAY{Jain\ \BBA\ Vazirani}{Jain\ \BBA\
  Vazirani}{2001}]{Jain2001}
Jain, K.\BBACOMMA\  \BBA\ Vazirani, V. \BBOP2001\BBCP.
\newblock \BBOQ Applications of approximation algorithms to cooperative
  games\BBCQ\
\newblock In {\Bem Proceedings of the 33rd annual ACM Symposium on Theory of
  Computing}, STOC '01, \BPGS\ 364--372, New York, NY, USA. ACM.

\bibitem[\protect\BCAY{Kalai\ \BBA\ Samet}{Kalai\ \BBA\
  Samet}{1983}]{Kalai1983}
Kalai, E.\BBACOMMA\  \BBA\ Samet, D. \BBOP1983\BBCP.
\newblock \BBOQ On weighted shapley values\BBCQ\
\newblock Discussion papers\ 602, Northwestern University, Center for
  Mathematical Studies in Economics and Management Science.

\bibitem[\protect\BCAY{Klijn}{Klijn}{2000}]{Klijn2000}
Klijn, F. \BBOP2000\BBCP.
\newblock \BBOQ An algorithm for envy-free allocations in an economy with
  indivisible objects and money\BBCQ\
\newblock {\Bem Social Choice and Welfare}, {\Bem 17\/}(2), 201--215.

\bibitem[\protect\BCAY{Krysta\ \BBA\ Ventre}{Krysta\ \BBA\
  Ventre}{2010}]{Krysta2010}
Krysta, P.\BBACOMMA\  \BBA\ Ventre, C. \BBOP2010\BBCP.
\newblock \BBOQ Combinatorial auctions with verification are tractable\BBCQ\
\newblock In {\Bem Proceedings of the 18th annual European Conference on
  Algorithms: Part II}, ESA'10, \BPGS\ 39--50, Berlin, Heidelberg.
  Springer-Verlag.

\bibitem[\protect\BCAY{Liben-Nowell, Sharp, Wexler,\ \BBA\ Woods}{Liben-Nowell
  et~al.}{2011}]{Liben-Nowell2011}
Liben-Nowell, D., Sharp, A., Wexler, T., \BBA\ Woods, K. \BBOP2011\BBCP.
\newblock \BBOQ Computing shapley value in cooperative supermodular games\BBCQ\
\newblock {\Bem Preprint available at
  http://www.oberlin.edu/faculty/kwoods/research/shapley.pdf}.

\bibitem[\protect\BCAY{Lindner}{Lindner}{2010}]{Lindner2010}
Lindner, C. \BBOP2010\BBCP.
\newblock \BBOQ A market-affected sealed-bid auction protocol\BBCQ\
\newblock In Konstantopoulos, S., Perantonis, S., Karkaletsis, V., Spyropoulos,
  C., \BBA\ Vouros, G.\BEDS, {\Bem Artificial Intelligence: Theories, Models
  and Applications}, \lowercase{\BVOL}\ 6040 of {\Bem Lecture Notes in Computer
  Science}, \BPGS\ 193--202. Springer Berlin / Heidelberg.

\bibitem[\protect\BCAY{Maniquet}{Maniquet}{2003}]{Francois2003}
Maniquet, F. \BBOP2003\BBCP.
\newblock \BBOQ A characterization of the shapley value in queueing
  problems\BBCQ\
\newblock {\Bem Journal of Economic Theory}, {\Bem 109\/}(1), 90--103.

\bibitem[\protect\BCAY{Maskin}{Maskin}{1987}]{Maskin1987}
Maskin, E. \BBOP1987\BBCP.
\newblock {\Bem On the Fair Allocation of Indivisible Goods}, \BPGS\ 341--349.
\newblock MacMillan.

\bibitem[\protect\BCAY{Meertens, Potters,\ \BBA\ Reijnierse}{Meertens
  et~al.}{2002}]{Meertens2002}
Meertens, M., Potters, J., \BBA\ Reijnierse, H. \BBOP2002\BBCP.
\newblock \BBOQ Envy-free and pareto efficient allocations in economies with
  indivisible goods and money\BBCQ\
\newblock {\Bem Mathematical Social Sciences}, {\Bem 44\/}(3), 223--233.

\bibitem[\protect\BCAY{Mishra\ \BBA\ Rangarajan}{Mishra\ \BBA\
  Rangarajan}{2007}]{Mishra2007}
Mishra, D.\BBACOMMA\  \BBA\ Rangarajan, B. \BBOP2007\BBCP.
\newblock \BBOQ Cost sharing in a job scheduling problem\BBCQ\
\newblock {\Bem Social Choice and Welfare}, {\Bem 29\/}(3), 369--382.

\bibitem[\protect\BCAY{Moulin}{Moulin}{1992}]{Moulin1992}
Moulin, H. \BBOP1992\BBCP.
\newblock \BBOQ An application of the shapley value to fair division with
  money\BBCQ\
\newblock {\Bem Econometrica}, {\Bem 60\/}(6), 1331--49.

\bibitem[\protect\BCAY{Moulin}{Moulin}{1999}]{Moulin1999}
Moulin, H. \BBOP1999\BBCP.
\newblock \BBOQ Incremental cost sharing: Characterization by coalition
  strategy-proofness\BBCQ\
\newblock {\Bem Social Choice and Welfare}, {\Bem 16\/}(2), 279--320.

\bibitem[\protect\BCAY{Moulin}{Moulin}{2003}]{Moulin2003}
Moulin, H. \BBOP2003\BBCP.
\newblock {\Bem Fair Division and Collective Welfare}.
\newblock MIT Press.

\bibitem[\protect\BCAY{Moulin\ \BBA\ Shenker}{Moulin\ \BBA\
  Shenker}{1997}]{Moulin1997}
Moulin, H.\BBACOMMA\  \BBA\ Shenker, S. \BBOP1997\BBCP.
\newblock \BBOQ Strategyproof sharing of submodular costs: budget balance
  versus effciency\BBCQ\
\newblock \BTR.

\bibitem[\protect\BCAY{Nagamochi, Zeng, Kabutoya,\ \BBA\ Ibaraki}{Nagamochi
  et~al.}{1997}]{Nagamochi1997}
Nagamochi, H., Zeng, D.-Z., Kabutoya, N., \BBA\ Ibaraki, T. \BBOP1997\BBCP.
\newblock \BBOQ Complexity of the minimum base game on matroids\BBCQ\
\newblock {\Bem Mathematics of Operations Research}, {\Bem 22}, 146--164.

\bibitem[\protect\BCAY{Nisan\ \BBA\ Ronen}{Nisan\ \BBA\
  Ronen}{2001}]{Nisan2001}
Nisan, N.\BBACOMMA\  \BBA\ Ronen, A. \BBOP2001\BBCP.
\newblock \BBOQ Algorithmic mechanism design\BBCQ\
\newblock {\Bem Games and Economic Behavior}, {\Bem 35}, 166--196.

\bibitem[\protect\BCAY{Ohseto}{Ohseto}{2004}]{Ohseto2004}
Ohseto, S. \BBOP2004\BBCP.
\newblock \BBOQ Implementing egalitarian-equivalent allocation of indivisible
  goods on restricted domains\BBCQ\
\newblock {\Bem Economic Theory}, {\Bem 23}, 659--670 (2004).

\bibitem[\protect\BCAY{Osborne\ \BBA\ Rubinstein}{Osborne\ \BBA\
  Rubinstein}{1994}]{Osborne1994}
Osborne, M.~J.\BBACOMMA\  \BBA\ Rubinstein, A. \BBOP1994\BBCP.
\newblock {\Bem A Course in Game Theory}.
\newblock The MIT Press, Cambridge, MA, USA.

\bibitem[\protect\BCAY{Papadimitriou}{Papadimitriou}{1993}]{Papadimitriou1993}
Papadimitriou, C.~H. \BBOP1993\BBCP.
\newblock {\Bem Computational Complexity}.
\newblock Addison-Wesley.

\bibitem[\protect\BCAY{Pathak}{Pathak}{2009}]{Pathak2009}
Pathak, A.P., S.~T. \BBOP2009\BBCP.
\newblock \BBOQ Comparing mechanisms by their vulnerability to
  manipulation\BBCQ\
\newblock \BTR, MIT.

\bibitem[\protect\BCAY{Penna\ \BBA\ Ventre}{Penna\ \BBA\
  Ventre}{2009}]{Penna2009}
Penna, P.\BBACOMMA\  \BBA\ Ventre, C. \BBOP2009\BBCP.
\newblock \BBOQ Optimal collusion-resistant mechanisms with verification\BBCQ\
\newblock In {\Bem Proceedings of the 10th ACM Conference on Electronic
  Commerce}, EC '09, \BPGS\ 147--156, New York, NY, USA. ACM.

\bibitem[\protect\BCAY{Porter, Shoham,\ \BBA\ Tennenholtz}{Porter
  et~al.}{2004}]{Porter2004}
Porter, R., Shoham, Y., \BBA\ Tennenholtz, M. \BBOP2004\BBCP.
\newblock \BBOQ Fair imposition\BBCQ\
\newblock {\Bem Journal of Economic Theory}, {\Bem 118\/}(2), 209 -- 228.

\bibitem[\protect\BCAY{Potthoff}{Potthoff}{2002}]{Potthoff2002}
Potthoff, R.~F. \BBOP2002\BBCP.
\newblock \BBOQ Use of linear programming to find an envy-free solution closest
  to the brams–kilgour gap solution for the housemates problem\BBCQ\
\newblock {\Bem Group Decision and Negotiation}, {\Bem 11}, 405--414.

\bibitem[\protect\BCAY{Quinzii}{Quinzii}{1984}]{Quinzii1984}
Quinzii, M. \BBOP1984\BBCP.
\newblock \BBOQ Core and competitive equilibria with indivisibilities\BBCQ\
\newblock {\Bem International Journal of Game Theory}, {\Bem 13}, 41--60.

\bibitem[\protect\BCAY{Sakai}{Sakai}{2007}]{Sakai2007}
Sakai, T. \BBOP2007\BBCP.
\newblock \BBOQ Fairness and implementability in allocation of indivisible
  objects with monetary compensations\BBCQ\
\newblock {\Bem Journal of Mathematical Economics}, {\Bem 43\/}(5), 549 -- 563.

\bibitem[\protect\BCAY{Schrijver}{Schrijver}{2003}]{Schrijver2003}
Schrijver, A. \BBOP2003\BBCP.
\newblock {\Bem Combinatorial Optimization: Polyhedra and Efficiency}.
\newblock Springer-Verlag.

\bibitem[\protect\BCAY{Shioura, Sun,\ \BBA\ Yang}{Shioura
  et~al.}{2006}]{Shioura2006}
Shioura, A., Sun, N., \BBA\ Yang, Z. \BBOP2006\BBCP.
\newblock \BBOQ Efficient strategy proof fair allocation algorithms\BBCQ\
\newblock {\Bem Journal of the Operations Research Society of Japan}, {\Bem
  49\/}(2), 144--150.

\bibitem[\protect\BCAY{Shoham\ \BBA\ Leyton-Brown}{Shoham\ \BBA\
  Leyton-Brown}{2009}]{Shoham2009}
Shoham, Y.\BBACOMMA\  \BBA\ Leyton-Brown, K. \BBOP2009\BBCP.
\newblock {\Bem Multiagent Systems}.
\newblock Cambridge University Press.

\bibitem[\protect\BCAY{Su}{Su}{1999}]{Su1999}
Su, F. \BBOP1999\BBCP.
\newblock \BBOQ Rental harmony: Sperner's lemma in fair division\BBCQ\
\newblock {\Bem American Mathematical Monthly}, {\Bem 106}, 930–942.

\bibitem[\protect\BCAY{Svensson}{Svensson}{1983}]{Svensson1983}
Svensson, L.-G. \BBOP1983\BBCP.
\newblock \BBOQ Large indivisibles: An analysis with respect to price
  equilibrium and fairness\BBCQ\
\newblock {\Bem Econometrica}, {\Bem 51\/}(4), pp. 939--954.

\bibitem[\protect\BCAY{Svensson}{Svensson}{2009}]{Svensson2009}
Svensson, L.-G. \BBOP2009\BBCP.
\newblock \BBOQ Coalitional strategy-proofness and fairness\BBCQ\
\newblock {\Bem Economic Theory}, {\Bem 40}, 227--245.

\bibitem[\protect\BCAY{Tadenuma\ \BBA\ Thomson}{Tadenuma\ \BBA\
  Thomson}{1991}]{Tadenuma1991}
Tadenuma, K.\BBACOMMA\  \BBA\ Thomson, W. \BBOP1991\BBCP.
\newblock \BBOQ No-envy and consistency in economies with indivisible
  goods\BBCQ\
\newblock {\Bem Econometrica}, {\Bem 59\/}(6), 1755--67.

\bibitem[\protect\BCAY{Tadenuma\ \BBA\ Thomson}{Tadenuma\ \BBA\
  Thomson}{1993}]{Tadenuma1993}
Tadenuma, K.\BBACOMMA\  \BBA\ Thomson, W. \BBOP1993\BBCP.
\newblock \BBOQ The fair allocation of an indivisible good when monetary
  compensations are possible\BBCQ\
\newblock {\Bem Mathematical Social Sciences}, {\Bem 25\/}(2), 117 -- 132.

\bibitem[\protect\BCAY{Tadenuma\ \BBA\ Thomson}{Tadenuma\ \BBA\
  Thomson}{1995}]{Tadenuma1995}
Tadenuma, K.\BBACOMMA\  \BBA\ Thomson, W. \BBOP1995\BBCP.
\newblock \BBOQ Games of fair division\BBCQ\
\newblock {\Bem Games and Economic Behavior}, {\Bem 9\/}(2), 191–204.

\bibitem[\protect\BCAY{Valiant}{Valiant}{1979a}]{Valiant1979}
Valiant, L.~G. \BBOP1979a\BBCP.
\newblock \BBOQ The complexity of computing the permanent\BBCQ\
\newblock {\Bem Theoretical Computer Science}, {\Bem 8\/}(2), 189--201.

\bibitem[\protect\BCAY{Valiant}{Valiant}{1979b}]{Valiant1979a}
Valiant, L.~G. \BBOP1979b\BBCP.
\newblock \BBOQ The complexity of enumeration and reliability problems\BBCQ\
\newblock {\Bem SIAM Journal on Computing}, {\Bem 8\/}(3), 410--421.

\bibitem[\protect\BCAY{Vazirani, Nisan, Roughgarden,\ \BBA\ Tardos}{Vazirani
  et~al.}{2007}]{Vazirani2007}
Vazirani, V.~V., Nisan, N., Roughgarden, T., \BBA\ Tardos, {\'E}.
  \BBOP2007\BBCP.
\newblock {\Bem Algorithmic Game Theory}.
\newblock Cambridge University Press, Cambridge, UK.

\bibitem[\protect\BCAY{Vickery}{Vickery}{1961}]{Vickery1961}
Vickery, W. \BBOP1961\BBCP.
\newblock \BBOQ Counterspeculation, auctions and competitive sealed
  tenders\BBCQ\
\newblock {\Bem Journal of Finance}, 8--37.

\bibitem[\protect\BCAY{William}{William}{2011}]{William2011}
William, T. \BBOP2011\BBCP.
\newblock \BBOQ Chapter twenty-one - fair allocation rules\BBCQ\
\newblock In Kenneth J.~Arrow, A.~S.\BBACOMMA\  \BBA\ Suzumura, K.\BEDS, {\Bem
  Handbook of Social Choice and Welfare}, \lowercase{\BVOL}~2 of {\Bem Handbook
  of Social Choice and Welfare}, \BPGS\ 393 -- 506. Elsevier.

\bibitem[\protect\BCAY{Willson}{Willson}{2003}]{Willson2003}
Willson, S.~J. \BBOP2003\BBCP.
\newblock \BBOQ Money-egalitarian-equivalent and gain-maximin allocations of
  indivisible items with monetary compensation\BBCQ\
\newblock {\Bem Social Choice and Welfare}, {\Bem 20}, 247--259.

\bibitem[\protect\BCAY{Yang}{Yang}{2001}]{Yang2001}
Yang, Z. \BBOP2001\BBCP.
\newblock \BBOQ An intersection theorem on an unbounded set and its application
  to the fair allocation problem\BBCQ\
\newblock {\Bem Journal of Optimization Theory and Applications}, {\Bem 110},
  429--443.

\bibitem[\protect\BCAY{Yengin}{Yengin}{2012}]{Yengin2012}
Yengin, D. \BBOP2012\BBCP.
\newblock \BBOQ Egalitarian-equivalent groves mechanisms in the allocation of
  heterogenous objects\BBCQ\
\newblock {\Bem Social Choice and Welfare}, {\Bem 38\/}(1), 137--160.

\bibitem[\protect\BCAY{Young}{Young}{1985}]{Young1985}
Young, H.~P. \BBOP1985\BBCP.
\newblock \BBOQ Monotonic solutions of cooperative games\BBCQ\
\newblock {\Bem International Journal of Game Theory}, {\Bem 14}, 65--72.

\bibitem[\protect\BCAY{Young}{Young}{1994}]{Young1994}
Young, H.~P. \BBOP1994\BBCP.
\newblock {\Bem Equity in Theory and Practice}.
\newblock Princeton University Press.

\end{thebibliography}

\end{document}